\documentclass{sig-alternate-10pt}

\usepackage{times}
\usepackage{graphicx}
\usepackage{subfigure}
\usepackage{url}
\usepackage{balance}
\newcommand{\comment}[1]{}

\newtheorem{theorem}{Theorem}[section]

\begin{document}
%
\title{Auctioning based Coordinated TV White Space Spectrum Sharing for Home Networks}

\numberofauthors{4}

\author{
\alignauthor Saravana Manickam and Mahesh K. Marina\\
      \affaddr{The University of Edinburgh}
\alignauthor Sofia Pediaditaki\\
       \affaddr{Intel}
\alignauthor Maziar Nekovee\\
       \affaddr{Samsung R \& D}
}

\maketitle

\begin{abstract}

The idea of having the geolocation database monitor the secondary use of TV white space (TVWS) spectrum and assist in coordinating the secondary usage is gaining ground. Considering the home networking use case, we leverage the geolocation database for interference-aware coordinated TVWS sharing among secondary users (home networks) using {\em short-term auctions}, thereby realize a dynamic secondary market. To enable this auctioning based coordinated TVWS sharing framework, we propose an enhanced {\em market-driven TVWS spectrum access model}. For the short-term auctions, we propose an online multi-unit, iterative truthful mechanism called \texttt{VERUM} that takes into consideration spatially heterogeneous spectrum availability, an inherent characteristic in the TVWS context. We prove that \texttt{VERUM} is truthful (i.e., the best strategy for every bidder is to bid based on its true valuation) and is also efficient in that it allocates spectrum to users who value it the most. Evaluation results from scenarios with real home distributions in urban and dense-urban environments and using realistic TVWS spectrum availability maps show that \texttt{VERUM} performs close to optimal allocation in terms of revenue for the coordinating spectrum manager. Comparison with two existing efficient and truthful multi-unit spectrum auction schemes, VERITAS and SATYA,  shows that \texttt{VERUM} fares better in terms of revenue, spectrum utilization and percentage of winning bidders in diverse conditions. Taking all of the above together, \texttt{VERUM} can be seen to offer incentives to subscribed users encouraging them to use TVWS spectrum through greater spectrum availability (as measured by percentage of winning bidders) as well as to the coordinating spectrum manager through revenue generation.

\end{abstract}
\maketitle
\section{Introduction}
\label{intro}


Home wireless use is on the rise. Given that people are typically indoors 80\% of the time and majority of that time is spent at home where wireless Internet access is prevalent, global consumer Internet traffic trends also reflect home wireless access trends. In particular, we can attribute the growing home wireless use to video traffic, which already dominates the global consumer Internet traffic and is expected to make up 80-90\% of the overall traffic by 2017. With the fiber based residential broadband access increasingly becoming commonplace, the bottleneck for home wireless Internet access will shift inside the home in future. Besides the conventional computing devices (e.g., PCs, laptops, tablets, smartphones), a whole host of other devices in the home are being equipped with wireless communications capability ranging from consumer electronics (e.g., TVs) and home appliances (e.g., refrigerators) to smart meters, leading to rise in {\em in-home wireless use} for multimedia distribution and interconnection among various home devices. While WiFi carries much of this home wireless traffic today, it is unlikely going to be sufficient going forward --- the 2.4GHz band is quite overcrowded and there are range concerns with the 5GHz band~\cite{maziar-mcom11}. Thus it is believed that meeting capacity and range requirements of future home wireless applications necessitate the use of other portions of the spectrum~\cite{dh-mcom12}.

In this context, the use of TV white space (TVWS) spectrum has been put forward as an additional complementary means to cope with the growing home wireless demand~\cite{maziar-mcom11,dh-mcom12,paws-usecases}. In simple terms, spectrum white spaces are unused portions of spectrum over space and time. TV white spaces then are white spaces in the UHF TV band (470-790MHz in the UK) as illustrated in Fig.~\ref{tvws}. Fig.~\ref{tvws} (a) also shows that digital terrestrial TV (DTT) transmitters and receivers, and program-making and special events (PMSE) devices are the primary users of the UHF TV spectrum. The parts of the spectrum which are unused by those primary users at a given location appear as spectrum white spaces as illustrated in Fig.~\ref{tvws} (b). Such white spaces can be used on an opportunistic basis by other secondary users provided no harmful interference is caused to existing primary users (i.e., DTT receivers and PMSE devices like wireless microphones). In fact, the TVWS spectrum is (being) opened up by regulators worldwide for unlicensed secondary use (sometimes also called secondary spectrum
commons) subject to interference protection for primary users. Moreover, there is now a general consensus among various regulatory proposals for TVWS use that geolocation database will be the mechanism to ensure primary user protection in the foreseeable future as other alternatives like spectrum sensing are unreliable or too expensive for the TV bands~\cite{webb-mcom12,maziar-mwc12}. With this mechanism, each secondary user (or its proxy) needs to first query a geolocation database providing its location to obtain the available TVWS channels. The geolocation database calculates its response based on the secondary user's location, DTT transmitter locations, propagation modelling and active PMSE users. Several TVWS oriented standards that use geolocation databases have already been published or in development (e.g., IETF PAWS, IEEE 802.19.1, IEEE 802.22, IEEE 802.11af, IEEE DySPAN-SC).

\begin{figure}
  \centering
  \subfigure[]{\includegraphics[width=2.7in]{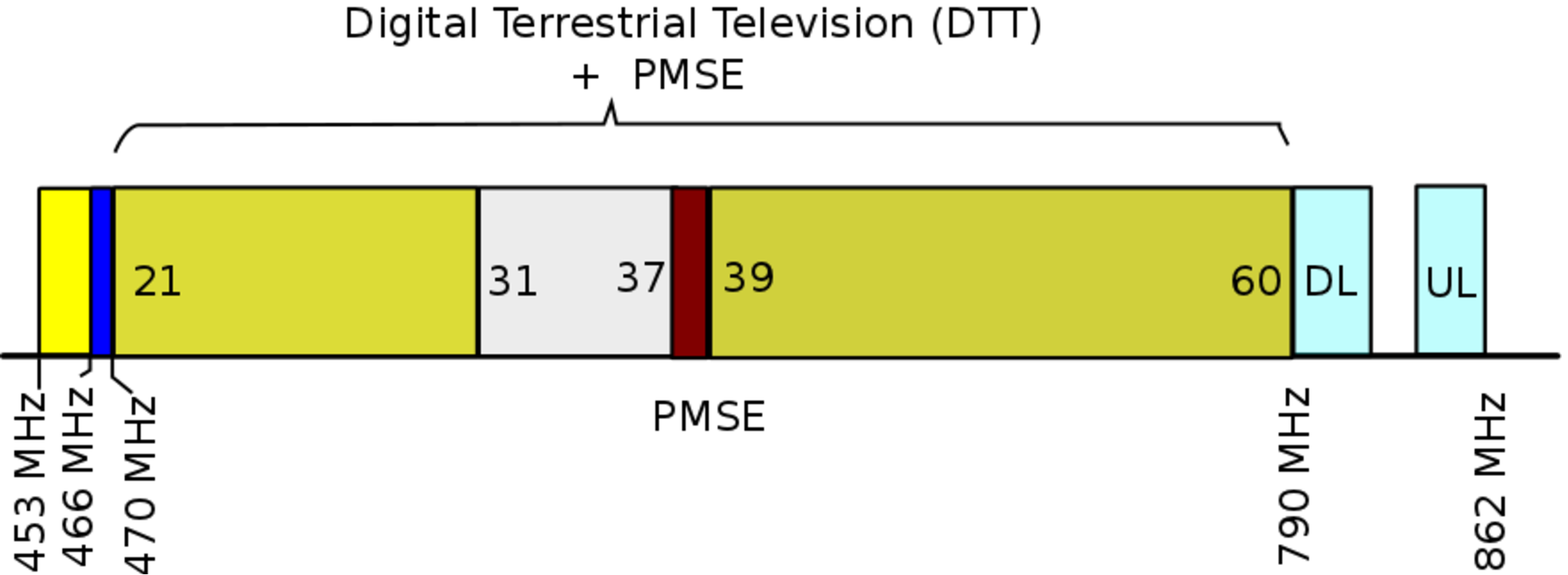}}
  \subfigure[]{\includegraphics[width=2.7in]{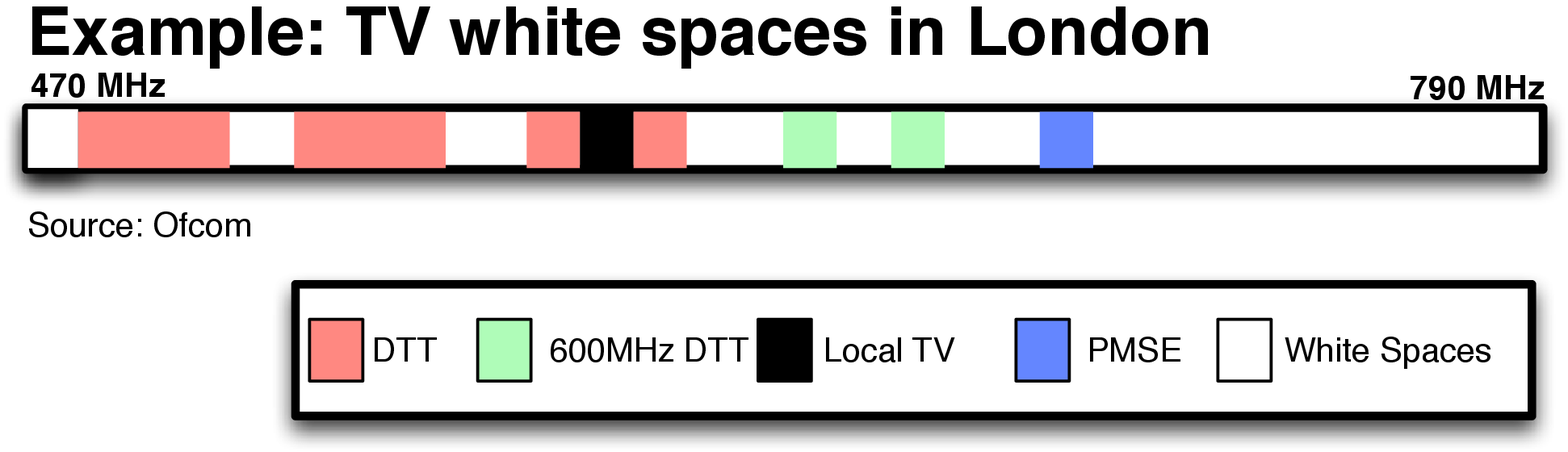}}
  \caption{(a) UHF TV spectrum in the UK where each channel is 8MHz wide; (b) TV white spaces in most of London.}
  \label{tvws}
\end{figure}

Home wireless networking makes a compelling use case for TVWS spectrum for three reasons: (1) it has superior propagation characteristics resulting in at least 4 times greater range compared to 2.4GHz band for the same power levels~\cite{ranveer-mc2r11}; (2) there is more spectrum available for low power use that is sufficient for indoor wireless applications --- for example, the recent TVWS related consultation from UK Ofcom~\cite{ofcom-sep13} shows that there are 30 channels available in central London at the low power level of 10dBm but increasing the power level to 35dBm drops that to 20 channels; (3) a recent measurement study~\cite{miser-mobicom13} shows that in urban areas up to 20\% more TVWS spectrum is available indoors compared to outdoors due to natural barriers in the form of walls.

The focus of this paper is on the key secondary coexistence problem that arises when considering the use of TVWS spectrum for home networking, which is to effectively share the available spectrum among secondary TVWS users (including home networks) to utilize it efficiently and also keep mutual interference low. Till date though, much of the focus on TVWS spectrum management justifiably has been on interference protection for primary users for the case of a single secondary user. Given that TVWS spectrum is unlicensed subject to primary protection via the geolocation database, it is in principle possible to take an uncoordinated approach for TVWS spectrum sharing like it is currently the case with WiFi. However, uncoordinated use via naive ALOHA like strategies risk unbounded interference among secondary users and consequently result in ineffective use of this new source of spectrum as demonstrated in Fig.~\ref{uncoordinated}. Even sophisticated alternatives following the uncoordinated approach (e.g., 802.11af\footnote{802.11af employs a similar medium access control (MAC) protocol like that in 802.11 (WiFi) standard but additionally has mechanisms to operate in TV bands and consult geolocation database prior to TVWS spectrum use.}) have limitations such as leading to starvation for some users when there are a heterogeneous set of secondary users~\cite{sumit-mwc11}. Note that TVWS spectrum is being considered for a wide range of uses from home networking and indoor hotspots to machine-to-machine (M2M) and public outdoor small cells~\cite{paws-usecases}. As such, heterogeneity of secondary users is inherent to TVWS spectrum.

In this paper, we take a coordinated approach to TVWS spectrum sharing among secondary users (including home networks) by leveraging the geolocation database. The need to consult a geolocation database before using TVWS spectrum naturally presents an opportunity to rely on the database not just for primary protection but also for coordinating spectrum usage among secondary users. This possibility has been alluded by the authors of the SenseLess geolocation database system~\cite{senseless}. The centralized TVWS spectrum coordination\footnote{One can draw parallels between our approach and cloud based control of public WiFi hotspots as well as cloud RANs.} also eases handling of heterogeneous secondary users. The idea of having the geolocation database monitor the secondary use of TVWS spectrum and assist in coordinating that usage is actually gaining ground as can be seen from \cite{paws-usecases,harada-mwc11} and the proprietary White Space Plus service~\cite{wsplus} provided by Spectrum Bridge, one of the approved database providers in the US and UK.

We present our coordinated TVWS spectrum sharing proposal in the context of home wireless networking use case. In particular, we introduce a centralized database assisted spectrum manager for home white space networks (HWSNs) that offers a coordination service for effective TVWS spectrum sharing among home networks with low mutual interference and efficient spectrum utilization. The spectrum manager's coordination service is enabled by a {\em market-driven TVWS spectrum access model and the use of short-term auctions} to dynamically (re-)allocate TVWS spectrum among home networks --- winners of the auction at a particular point in time can be seen as securing a temporary permit to use a portion of the TVWS spectrum. Note that the key ideas of market-driven spectrum access model and use of short-term auctions underlying our proposal are more generally applicable to other spectrum bands and use cases.

\comment{
 novelty:
}

Specifically, we make the following contributions:

\begin{itemize}
\item We propose an auctioning based framework for coordinated TVWS spectrum sharing among HWSNs. To enable this framework, we present a market-driven TVWS spectrum access model as an enhancement to the current secondary spectrum commons model. With this new market based access model, every HWSN is a subscriber to the TVWS coordination service offered by a HWSN spectrum manager, which could be same as the database provider in practice. HWSN spectrum manager relies on the geolocation database for monitoring and updating secondary TVWS spectrum usage, and at the same time dynamically allocates TVWS spectrum to HWSNs using short-term auctions.
\item We develop \texttt{VERUM}\footnote{VERUM means truth in Latin.}, an online multi-unit auction mechanism to realize the HWSN spectrum manager's coordination service as mentioned above. The primary design objective behind \texttt{VERUM} is to achieve an efficient allocation, i.e., allocate spectrum to users (home networks) who value it the most; revenue/profit maximization for the spectrum manager (auctioneer) is seen as a secondary goal. We prove that \texttt{VERUM} auction mechanism yields efficient allocation, and that it is truthful thereby simplifying the user side behaviour in the coordination process. In comparison with existing efficient and truthful multi-unit spectrum auction mechanisms~\cite{veritas,satya,small,Hoefer-comm12}, \texttt{VERUM} takes an iterative approach that is fundamentally different and offers a simpler means to ensuring truthfulness. \texttt{VERUM} handles heterogeneous spectrum availability and demands, inherent to TVWS based home networking. It also allows channel sharing among nearby home networks subject to a specified interference limit. Moreover, \texttt{VERUM} supports marginal valuations, which refer to the values a bidder associates with the first channel in its demand and every additional channel; marginal valuations aid in achieving improved spectrum utilization and revenue by facilitating flexible adaptation of satisfied demand depending on spectrum availability.
\item We evaluate \texttt{VERUM} using realistic TV white space availability maps in the UK and actual distribution of homes in urban and sub-urban environments. Our results show that \texttt{VERUM} performs close to optimal allocation in terms of revenue. We also show that its performance in terms of revenue, spectrum utilization and percentage of winning bidders in the auction is better compared to VERITAS~\cite{veritas} and SATYA~\cite{satya}, other comparable efficient and truthful online spectrum auction mechanisms from the literature, as a result of its ability to support channel sharing, heterogeneous spectrum availability and marginal valuations.
\end{itemize}

The rest of the paper is structured as follows. The next section discusses related work. Our system model, including the proposed auctioning based HWSN spectrum coordination framework, is described in section 3. In section~\ref{auctions}, we detail \texttt{VERUM}, the online truthful and efficient multi-unit auction mechanism we propose in this paper. Evaluation results for \texttt{VERUM} are presented in section 5. Discussion of further issues is provided in section 6 and we conclude the paper in section 7.

\section{Related Work}
\label{related}

\subsection{TVWS Research}
Since FCC allowed unlicensed operation in TV bands in November 2008, there has been a significant amount of research activity focused on TV white spaces. These include TVWS spectrum availability assessments in different regions (e.g., \cite{sahai-dyspan10} for US, \cite{mahonen-dyspan11} for Europe). Spectrum usage measurement studies have also been carried out in different parts of the world to assess the potential for opportunistic access by secondary users in various licensed portions of the wireless communications spectrum, including the UHF spectrum with TV white spaces (e.g., \cite{mingyan-tmc12,heather-ton12}). Platforms based on networked spectrum analyzers for distributed spectrum sensing have also been developed for real-time spatio-temporal usage maps for TVWS spectrum and other applications (e.g., SpecNet~\cite{specnet-nsdi11}).

The pioneering work of Bahl et al.~\cite{whitefi} identifies salient aspects of TVWS spectrum in comparison with WiFi --- spatio-temporal variation and spectrum fragmentation, and implements the WhiteFi system prototype that incorporates techniques to handles these aspects. The same team has subsequently deployed the first operational white space network at Microsoft Redmond campus based on WhiteFi~\cite{ranveer-mc2r11}. While the focus of WhiteFi is on a single access point (AP) with multiple associated clients, Deb et al.~\cite{deb-mobicom09} consider a multi-AP infrastructure wireless LAN scenario and address the load-aware white space spectrum distribution, extending the previous work on the same problem for WiFi by Moscibroda et al.~\cite{moscibroda-icnp08}. Progress has also been made on other core technologies needed for the use of TVWS spectrum, including: geolocation database systems (e.g., \cite{senseless}); simultaneous use of spectrum fragments or non-contiguous portions of TVWS spectrum (e.g., \cite{wifinc}); and addressing MAC and link adaptation issues in white space networks (e.g., \cite{bozidar-conext11}).

Concerning the problem tackled in this paper --- TVWS spectrum sharing among secondary users, it is sometimes subsumed under a more general term called coexistence.
As the geolocation database is a common entity that needs to be consulted anyway to avoid harmful interference to primary users (e.g., DTV receivers), it makes a natural candidate to additionally provide a common store for monitoring secondary TVWS usage and thereby assist in coordinated TVWS spectrum use. This possibility has been accommodated in recent works~\cite{senseless,paws-usecases}. Villardi et al.~\cite{harada-mwc11} take it a step further and use this capability as a part of their centralized coexistence mechanism which involves TVWS access points making independent decisions on secondary use based on the TVWS spectrum availability information (including other secondary use) retrieved from the geolocation database. However, they only outline the scheme without providing details on how channel selection is made. In fact, the authors in \cite{harada-mwc11} themselves state that their focus is on providing a coexistence framework and not on the actual decision making mechanism. Somewhat similar to \cite{harada-mwc11} is IEEE 802.19.1\footnote{\url{http://ieee802.org/19/pub/TG1.html}}, a radio technology independent standard that is currently under development for coexistence among dissimilar TV Band Devices (TVBDs) and dissimilar or independently operated networks of TVBDs.


%
\subsection{Auctions for Dynamic Spectrum Access}

Auctions have been extensively used over the years for dynamic spectrum management with several different types of auction mechanisms developed to suit different scenarios~\cite{jordan-mcom11,dusit-comst12}. In the TVWS context, the work by Bogucka et al.~\cite{bogucka-mcom12,bogucka-dyspan11} is the closest to ours in that they also employ short-term auctions for secondary use of TVWS spectrum. However, there are several crucial differences. Firstly, they target a very different use case of secondary licensing of TVWS spectrum to mobile cellular operators (LTE) for offloading during peak traffic periods. As a consequence, there are a small number of bidders (operators) allowing them to rely on computationally expensive ``branch-and-cut'' optimization method. In contrast, in our problem setting there could be a large number of home networks participating in the auction. Secondly, the objective of the auction in \cite{bogucka-mcom12,bogucka-dyspan11} is to maximize revenue and they employ an untruthful first price
sealed bid auction to achieve that objective. They also place substantial focus on reserve price estimation given their revenue maximization goal. This is unlike our approach where we primarily target an efficient auction by allocating spectrum to those users who value it the most. As another difference, there is no notion of channel sharing among different users in \cite{bogucka-mcom12,bogucka-dyspan11}. The above discussion further highlights the fact that different use cases and objectives necessitate different auction mechanisms.

Outside the TVWS context, there are other auction mechanisms in the literature for dynamic spectrum management based on cognitive radios that could be adapted to our problem setting. In the following, we mainly discuss the auction mechanisms from previous work that are most relevant from this perspective, i.e., the ones which are suitable for auctioning multiple objects while at the same time efficient and truthful.

Zhou et al.~\cite{veritas} propose VERITAS, a truthful online auction scheme for dynamic spectrum management. A bidder must have the same valuation for all the channels and this is reflected in the pricing function, where every channel allocated to the bidder is charged at the same price. In other words, VERITAS does not support marginal valuations. It also does not support channel sharing, where a channel could be shared by neighboring home networks. Finally, VERITAS does not support heterogeneous channel availability, an inherent characteristic of TVWS networks. Modifying VERITAS to support marginal valuations, channel sharing and heterogeneous channel availability, however, is non-trivial.

Kash et al.~\cite{satya} propose SATYA, a truthful auction scheme for spectrum sharing that uses bucketing and ironing of bids to maintain monotonicity for truthfulness. While this is the first scheme to support channel sharing it has a few drawbacks. Firstly, it does not support marginal valuations. Secondly, it has an exponential run time and is only polynomial under some restrictions. Finally, the performance of the scheme is highly dependant on the bucketing function which is not part of the SATYA mechanism and is abstracted out.

Wu et al.~\cite{small} propose a truthful auction mechanism SMALL for non-cooperative wireless networks. They group bidders into non-conflicting groups and use a pricing function that is independent of the buyers' bid. However it is not suitable for online auctions because for SMALL to be truthful, it needs the optimal solution for graph coloring, known to be NP-complete. Moreover, SMALL does not support marginal valuations, channel sharing and heterogeneous channel availability.

Hoefer and Kesselhiem~\cite{Hoefer-comm12} study truthful spectrum auctions for secondary markets. They propose randomized allocation algorithms for weighted and unweighted conflict graphs that yield truthfulness in expectation. Truthfulness is achieved using the randomized meta-rounding framework~\cite{lavi-swamy-2005}. While the auction schemes offer near optimal guarantees on social welfare, they are not suitable for online operation. The convex optimization techniques that are used for applying randomized meta-rounding have runtime that exponentially grows with number of nodes in the conflict graph used to model interference relationships among different users. 

To summarize, qualitatively speaking, in comparison with the auction mechanisms discussed above, the novelty of \texttt{VERUM} proposed in this paper is that it is an online scheme that supports marginal valuations, channel sharing and heterogeneous channel availability. While SATYA also supports channel sharing, it is only polynomial under certain restrictions. Moreover, as discussed later in section~\ref{auctions}, \texttt{VERUM} takes an iterative approach that is fundamentally different from the other mechanisms and leads to not only simpler means to achieve truthfulness but also offers other desirable properties like transparency and privacy.

\section{System Model}
\label{model}

\begin{figure}[h!]
  \begin{center}
   \includegraphics[width = .6\columnwidth]{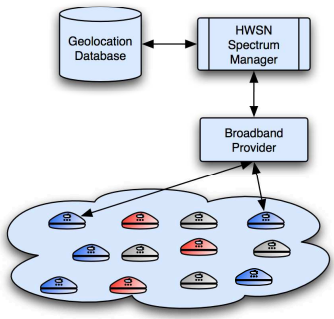}
    \caption{The database assisted home white space network (HWSN) spectrum manager illustrated. Each triangle in the cloud corresponds to a potential secondary user in the form of a home network. The ones colored in blue are HWSNs that access the TVWS spectrum via the HWSN spectrum manager shown, whereas the red colored ones are considered as third party white space networks (WSNs) that access the TVWS spectrum via other spectrum managers. Grey colored triangles represent homes that do not use TVWS spectrum at a given point in time.}
\label{framework}
 \end{center}
\end{figure}



\subsection{Auctioning based Coordinated TVWS Spectrum Sharing Framework}

In this paper, we are interested in TVWS spectrum sharing among secondary users from the perspective of home wireless networking. As such, home networks using TVWS spectrum are secondary users of specific interest to us and we refer to them as home white space networks (HWSNs). Within each HWSN, TVWS spectrum access for in-home white space devices (WSDs) is via the corresponding home TVWS access point (AP). In other words, the TVWS AP acts as the master device and obtains the TVWS channel availability information from the geolocation database on behalf of its associated slave WSDs using the home broadband Internet connection to connect to the database\footnote{The use of broadband connection to consult the database makes the bootstrapping problem discussed in \cite{senseless} a non-issue in our setting.}. Following \cite{ofcom-db}, we assume that each AP includes its location accuracy based on the expected distance to its furthest WSD when querying the geolocation database to obtain available TVWS channels for its WSDs. We work at the level of HWSNs and assume that techniques similar to \cite{whitefi,wifinc,bozidar-conext11} are employed within each HWSN for MAC layer operation and aggregation of non-contiguous channels.

For effective TVWS spectrum sharing among home networks, we introduce an entity called {\em HWSN Spectrum Manager} that is tasked with coordinating TVWS spectrum usage among HWSNs. Specifically, the spectrum manager allocates TVWS channels to HWSNs by taking into consideration the TVWS spectrum availability and secondary usage information from the geolocation database, spectrum demand information from each HWSN and mutual interference relationships among HWSNs. In the spirit of \cite{senseless} and \cite{paws-usecases}, we assume that geolocation database is continually updated with secondary TVWS spectrum usage. Fig.~\ref{framework} illustrates this database assisted TVWS spectrum coordination framework for home networks.

In our proposal, HWSN spectrum manager coordinates TVWS spectrum sharing among HWSNs via {\em short-term auctions}. This would mean that the spectrum manager plays the role of the auctioneer while HWSNs act as bidders (for accessing TVWS spectrum). Unlike the conventional long-term nation-wide spectrum auctions, short-term auctions allow HWSNs to share available TVWS spectrum over space and time while keeping interference between them tolerable.

\subsubsection{Market-Driven TVWS Spectrum Access Model}

As noted at the outset, unlicensed access by secondary users (secondary spectrum commons) with the requirement to consult the geolocation database first (to ensure primary protection) is the approach being adopted by regulators currently. Such an approach is sufficient when following an uncoordinated approach to share TVWS spectrum. However it undermines the coordinated approach in an environment where there are a mix of coordinated and uncoordinated users. For example, bootstrapping the coordinated service by incentivizing uncoordinated users to take up the service becomes an issue. Also in a mixed environment, users who fail to access TVWS spectrum via the coordinated approach at a given point in time can attempt the uncoordinated approach as a fallback and thereby hurt other users accessing the spectrum via the coordination service.

Thus to enable auctioning based coordinated TVWS spectrum sharing, we envision an enhanced TVWS spectrum access model which is {\em market-driven}. In this enhanced model, different HWSNs (or more generally, users of TVWS spectrum) become players in the market who each place a value on the TVWS spectrum that could potentially vary between them and over time. They rely on a spectrum manager to obtain access to TVWS spectrum. At any given point in time, the role of the spectrum manager then is to allocate the TVWS spectrum to users who value it the most while taking into consideration mutual interference relationships among them.

The implementation of the market-driven spectrum access model outlined above is straightforward from a regulatory viewpoint given that as per current regulations the database needs to be consulted as a first step prior to TVWS spectrum use. All that the database / spectrum manager needs to do additionally by way of implementing the market-driven model is to enforce priority to subscribers of the coordination service and use a protocol that involves TVWS users to allocate TVWS spectrum dynamically via short-term auctions. For new white space spectrum bands in future, market-driven access could be mandated by the regulator as the only means of access just as some bands offer licensed or unlicensed access currently.

The market based coordinated access model can be seen as an intermediate one between the extremes of licensed access and unlicensed access; it aims to combine the advantages of both licensed and unlicensed models while avoiding their disadvantages. It can also be viewed as a way of realizing a dynamic secondary market via regulator designated band managers (e.g., database providers). There are other recent proposals for alternate spectrum access models that also conceptually sit between (secondary) spectrum commons and licensed access such as pluralistic licensing~\cite{plural} and licensed shared access~\cite{lsa-wp}. While pluralistic licensing aims to incentivize primary users to allow more spectrum to be available under the secondary commons model, market based access enhances secondary commons via coordination that prioritizes users based on their relative spectrum valuations; in that sense, both pluralistic licensing and market based access are complementary. Unlike licensed shared access which is based on secondary exclusive licenses, market based access can be potentially more efficient because it is effectively a soft/dynamic licensing model. We leave further discussion and comparison among various spectrum access models for another paper.

\subsubsection{Business Models for TVWS Coordination Service}

Within the context of the above market based spectrum access model, every HWSN is a subscriber of coordination service offered by a HWSN spectrum manager. We now briefly discuss the business model aspect of such a coordination service. One possibility is that TVWS access can be an add-on over the user's monthly broadband subscription fee for a small additional fee. Depending on their TVWS related fee, subscribers could fall into different classes (e.g., ``gold'', ``silver'', ``bronze'') with correspondingly different number of monthly credits that put a cap on TVWS spectrum usage (see section III.A.5 for a brief discussion on how monthly credits can play a part when user attempts to use TVWS spectrum). Alternatively, subscribers for the coordination service could be at the level of broadband providers who in turn manage allocation of TVWS spectrum and payment internally with their subscribers.

The HWSN spectrum manager, while conceptually different from the database provider as shown in Fig.~\ref{framework}, can in practice be the same as the database provider and offer the coordination service for additional revenue generation. Moreover, there may be more than one spectrum manager in practice providing the coordination service just as there would be several database providers. In such a case, TVWS spectrum users will subscribe to one among the several managers. We assume that the different spectrum managers constantly synchronize the secondary use of TVWS spectrum with each other.

\subsubsection{Interference Modelling}

Since potential interference among TVWS spectrum users is a key factor to consider in the spectrum allocation, we need to explicitly model interference between different HWSNs. For this, we use the commonly used {\em conflict graph} model to represent interference relationships among HWSNs. Specifically, vertices in the conflict graph $G$ correspond to individual HWSNs. An edge exists between two vertices in $G$ if the corresponding HWSNs can interfere when using the same TVWS channels. Several alternatives are possible to infer such potential interference between a pair of HWSNs, from the simpler protocol model to the more sophisticated physical model with realistic propagation models~\cite{jain-mobicom03}. In this paper, we assume symmetric interference relationships that results in an undirected conflict graph.

Fig.~\ref{FigExample1} shows an example conflict graph for a set of HWSNs A-E within a given area. Interfering HWSNs are connected by a line. Note that conflict graph in general may not be a connected graph. In other words, it can be composed of several connected components.


We use notation $N_i$ to refer to the set of nodes to which vertex $i$ (a HWSN) is connected in the conflicted graph $G$. Note that our auctioning based spectrum allocation framework allows for sharing of TVWS channels among potentially interfering HWSNs (i.e., neighbors in $G$) subject to a pre-specified interference temperature limit. This is discussed further in section \ref{sharingsection}.

\subsubsection{Heterogeneous Spectrum Availability}

As incumbent TV broadcasters cover a large area, they may not cause heterogeneous availability of channels within a small area. PMSE (Program Making and Special Events) users, on the other hand, may occupy some channels within a small area thus could potentially cause differences in channels available over space and time. Third party white space networks (WSNs), colored red in Fig.~\ref{framework}, that are subscribers of other HWSN spectrum managers, will likely be the key reason behind heterogeneous spectrum availability across different HWSNs. Third party WSNs access TVWS spectrum through other HWSN spectrum managers and can therefore impact the spectrum availability for other secondary users.

For example, consider the conflict graph shown in Fig.~\ref{FigExample1}. Suppose that HWSN E is accessing the TVWS spectrum through another HWSN spectrum manager, then the channels used by HWSN E would be unavailable only for HWSN D while the spectrum availability for other HWSNs A, B and C may remain unaffected.


\subsubsection{Spectrum Demands and Marginal Valuations}

We assume that there is a mechanism available at each home TVWS AP to translate aggregate throughput demands from all in-home WSDs into an overall HWSN spectrum demand for the home in terms of number of TVWS channels. For example, for a given spectral efficiency, the bit-rate achievable for a channel with a certain bandwidth can be determined (e.g., with 1bit/second/Hz spectral efficiency, a 8MHz channel would result in 8Mbps bit-rate.). The achievable bit-rate so computed can be combined with the application bandwidth requirement to deduce the number of channels needed. Continuing the example, if we consider a HDTV streaming application with 20Mbps bandwidth requirement, then accounting for losses in efficiency due to aggregation of non-contiguous channels we need 3 channels each capable of 8Mbps bit-rate. In the case of channel sharing, the spectrum demand can include a fraction of a channel (see section \ref{sharingsection}).

We further assume that there is a mechanism within each bidder (HWSN) to generate private marginal valuations as a vector of length equal to the number of requested TVWS channels. Marginal valuations of a bidder refer to values that the bidder associates with the first channel in its demand and every additional channel. While aggregated bandwidth needs of all in-home WSDs determine the spectrum demand (number of TVWS channels requested) of a HWSN, other factors like number of credits the user is willing to spend at a given point in time, the set of individual applications driving the demand and their service requirements (best effort vs. guaranteed service) will influence the choice of marginal valuations. We assume that marginal valuations of each bidder are {\em weakly decreasing}. Fig.~\ref{FigExample1} shows example marginal valuations (with ``V:'' above each HWSN). Considering HWSN A in Fig.~\ref{FigExample1} as a specific example, A's marginal valuations show that it values the first channel it can get
at 13, the second channel at 8 and a third channel at 6.

\subsection{Problem Statement}

Broadly speaking, the goal of the HWSN spectrum manager in our auctioning framework is to allocate TVWS channels to each actively participating HWSN $i$ in a set of $n$ HWSNs, each with spectrum demand, $D_i \le C$, where $C$ is the total number of TVWS channels across all HWSNs that can be allocated from HWSN spectrum manager's perspective.

Recall that $N_i$ represents the set of neighbors of HWSN $i$. The aggregate demand of HWSN $i$'s neighbours is represented by $D_{-i} = \sum_{j \in N_i} D_j$.

To represent the spectrum available at HWSN $i$, we use a bit vector $X_i$ of size $C$. We refer to this vector as the {\em channel availability vector} defined as

 \begin{equation}
  X_i(k) = \left\{ \begin{array}{ll}
         1 & \mbox{if channel $k$ is available at $i$};\\
         0 & \mbox{otherwise}.\end{array} \right. \\
 \end{equation}

 We define the {\em channel assignment vector} $Y_i$ at $i$ as,

 \begin{equation}
  Y_i(k) = \left\{ \begin{array}{ll}
         1 & \mbox{if channel $k$ is assigned to i};\\
         0 & \mbox{otherwise}.\end{array} \right.
 \end{equation}

 We also define the number of channels available at HWSN $i$ as $x_i = \sum_{k=1}^{C} X_i(k)$, and the number of channels assigned to $i$ as $y_i = \sum_{k=1}^{C}  Y_i(k)$.

%
%

As already stated in section 3.1, in our framework the HWSN spectrum manager coordinates the allocation of the available TVWS spectrum among active HWSNs (secondary users) in an interference-aware manner each time via a short-term auction.

The primary design objective for our proposed auctioning mechanism \texttt{VERUM} described in the next section is that it should lead to efficient allocation, i.e., allocate spectrum to HWSNs that value it the most. Other key considerations underlying our design are listed below.

\begin{itemize}


\item {\em Truthful or Strategy-Proof}: For every bidder, bidding based on its true valuation should be its best strategy. This eliminates the possibility of bidding
strategies that can affect the outcome of the auction. In our context, the biggest advantage of a strategy-proof auction is the simplification of bidder's
dominant strategy. This facilitates implementation of the coordination mechanism among HWSNs with less complexity.

 \item {\em High Revenue}: While this is not our primary goal, revenue generation serves as an incentive for HWSN spectrum manager (the auctioneer) to run the secondary spectrum use coordination service.

 \item {\em Low Computational Complexity}: This is needed for real-time allocation and re-allocation of TVWS spectrum. It is challenging to meet with large number of HWSNs with diverse time-varying spectrum demands and channel availability bidding for spectrum in the auction.

\item {\em Efficient Spectrum Utilization}: This is the main motivation behind dynamic spectrum access in general. As such it is also an important goal for our work on facilitating effective TVWS spectrum use among home networks.

\end{itemize}

\section{VERUM Auction Mechanism}
\label{auctions}

\subsection{Overview}
The overall architecture of the system based on our proposed auction mechanism \texttt{VERUM} is shown in Fig.~\ref{auctionmodel} (a). \texttt{VERUM} is an online iterative multi-unit auction mechanism that spans multiple rounds. \texttt{VERUM} supports marginal valuations, channel sharing and heterogeneous spectrum availability. We also prove later in this section that it is truthful and efficient.

The timeline for the system in operation is shown in Fig.~\ref{auctionmodel} (b). Time is seen as a sequence of {\em epochs}, each consisting of a short {\em Auction Phase} followed by a much longer {\em Spectrum Use Phase}.  Each {\em Auction Phase} consists of one or more {\em rounds} involving interaction between the auctioneer (HWSN spectrum manager) and bidders (participating HWSNs) as part of the auction to meet the spectrum demand of the bidders subject to their valuations, spectrum availability and mutual interference relationships.  At the beginning of auction phase, the auctioneer announces the initial price, i.e., the {\em reserve price}. It then waits for a {\em bidding period} to receive demands from HWSNs. Depending on the demands, the auctioneer may allocate one or more channels to some of the bidders at the current round price. The auction then may also proceed to another round by increasing the reserve price to bring down excess demand. This process may continue over several rounds until there there is no more demand to be fulfilled. At the end of the auction phase, HWSNs whose bids are successful proceed to use TVWS channels they won in the following Spectrum Use phase until the end of that epoch. Same process repeats in the next epoch and so on.

\begin{figure}
  \begin{center}
  \begin{tabular}{c}
   \subfigure[]
{
	\includegraphics[width=2.4in]{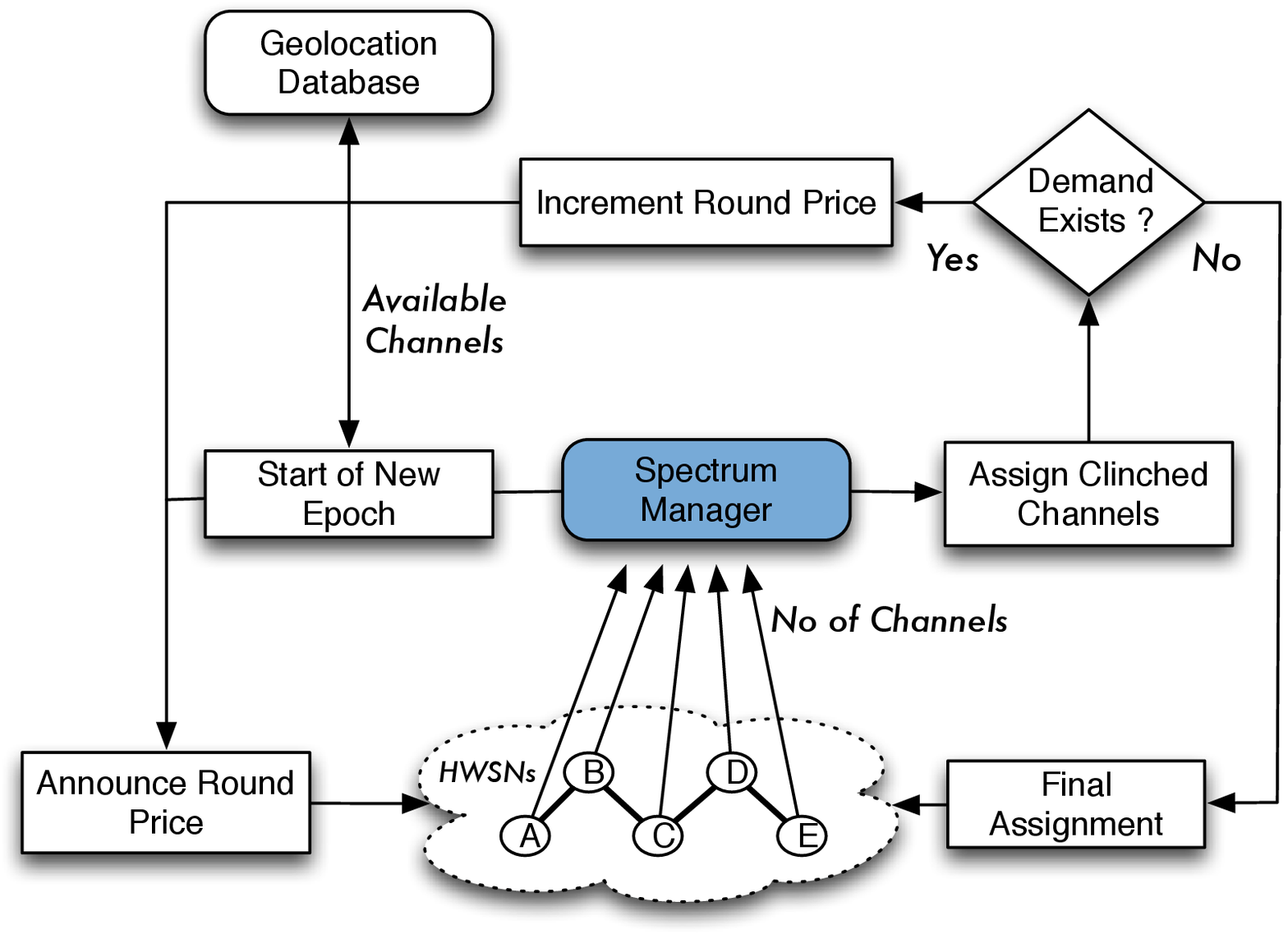}
 }
\\
   \subfigure[]
{
	\includegraphics[width=2.4in]{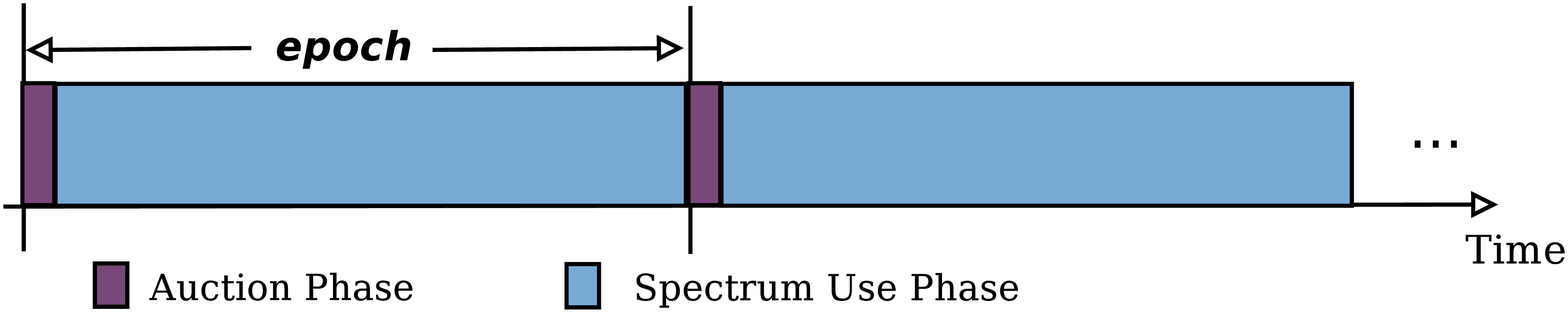}
 }
 \end{tabular}
\end{center}
 \caption{(a) System architecture with VERUM (b) Timeline of the auctioning based coordinated TVWS spectrum access system.}
  \label{auctionmodel}
\end{figure}

In the following, \texttt{VERUM} is described in two steps starting with the simpler case of exclusive channel use where a channel is allocated only to a set of mutually non-interfering HWSNs. In section \ref{sharingsection}, we build on the exclusive use model to support channel sharing among HWSNs.

\subsection{The Case of Exclusive Channel Use}



\subsubsection{Existing Auction Designs}

We begin by considering alternative designs to motivate our auction design. In our problem, individual bidders may request multiple items (channels) with private and independent valuations. Moreover, we seek an efficient and truthful auction. Both of these would suggest the use of multi-unit Vickrey auction proposed in the seminal paper~\cite{vickery1961}. In the classical multi-unit Vickrey auction, there are $K$ items to be sold at the auctioneer and bidders submit sealed bids. Winner determination is straightforward: $K$ highest bids are deemed winning bids. The price that each winner pays is more involved in that a bidder who wins $M$ items ($M$ $\le$ $K$) pays the opportunity cost for those $M$ items. Specifically, the bidder winning $M$ items will pay the amount of the $M^{th}$ highest losing bid for the first item, $(M-1)^{th}$ highest losing bid for the second item and so on.

When we apply the above multi-unit Vickrey auction mechanism in the spectrum allocation context where spatial reuse is allowed for efficient spectrum utilization and conflict (interference) relationships need to be accounted, computing Vickrey pricing described above becomes complex as each bidder could be allocated a channel and there may as such be no losing bid. We could simplify the pricing scheme but that can come at the expense of truthfulness as illustrated in the rest of this paragraph. Suppose channels are assigned in a non-conflicting way starting from the highest bidder until all the bidders are considered or the channels are exhausted. Then consider the following pricing scheme: the winner $i$ is charged the highest bid of the unallocated conflicting neighbor if there exists one, otherwise it is charged zero. However such as a scheme violates truthfulness. To see this, consider the conflict graph shown in Fig.~\ref{prooffig1} where HWSNs $A$, $B$, $C$, and $D$ are bidding for one channel each and suppose that there are two channels ($C_1$ and $C_2$) available at all the HWSNs. When the bidders bid truthfully (i.e., at their valuations), it can be seen that the utilities are 5, 3, 0, and 1 with $C_1$ assigned to $A$ and $D$, and $C_2$ assigned to bidder $B$ at prices 0, 1, and 1 respectively. If however bidder $C$  strategically bids 3 (exceeding its valuation and thus untrue) with resulting increase in its utility to 1, all the bidders are now charged zero (or some reserve price). An implication of a truthful auction scheme is that the best strategy for each bidder is to bid truthfully. By showing that the best strategy for bidder $C$ is not to bid truthfully we show that the auction scheme is not truthful. This is also discussed by Zhou et al.~\cite{veritas}.

  \begin{figure}[h]
  \centering
  \includegraphics[width=2.8in]{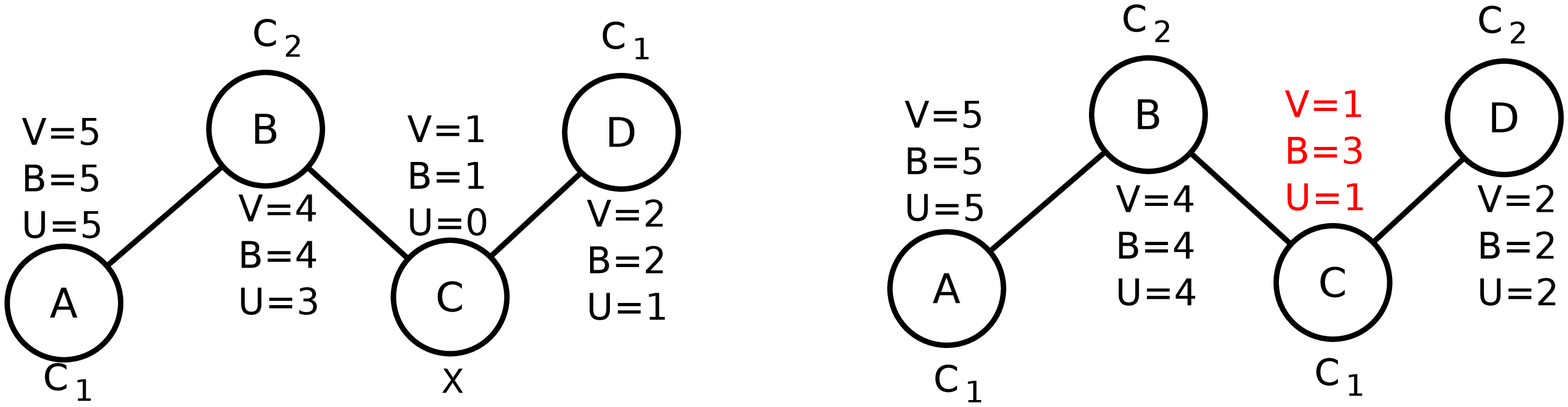}
  \caption{Example which illustrates that variants of Vickrey auction with simplified pricing can lead to untruthful allocation when applied to spectrum allocation. Left side of figure corresponds to the truthful bidding case while right one is for strategic bidding. In both cases, V stands for bidder's valuation for obtaining a channel, B is the bid amount and U denotes the bidder's utility calculated as difference between its valuation and the price it pays.}
  \label{prooffig1}
  \end{figure}
%

To ensure truthfulness, existing multi-unit spectrum auction schemes~\cite{veritas,satya,small,Hoefer-comm12} employ different means to realize Vickrey pricing and end up being polynomial only under certain constraints or have an exponential run time. All these auction mechanisms are variants of the sealed bid auction scheme where the bidders submit sealed bids and the auctioneer computes the winners based on all the bids. The price that the winning bidders have to pay must enable truthfulness (e.g., via Vickrey or VCG pricing).

In VERITAS~\cite{veritas}, the auctioneer collects sealed bids from all the bidders and uses a greedy algorithm for allocation. It then determines a critical neighbor for each of the bidders based on which the winning price is computed. To support channel sharing, the critical neighbor determination scales exponentially with the number of neighbors for each winning bidder.

In SATYA~\cite{satya}, the auctioneer collects sealed bids from all the bidders and uses a greedy algorithm for allocation. During allocation, buckets are used to prevent non-monotonic allocations and ironing is used to remove all possible non-monotonic allocations. Note that an auction scheme is monotonic if a larger bid from a bidder always leads to at least as many channels being allocated as with a smaller bid. Monotonicity is known to be a necessary and sufficient condition for an auction scheme to be truthful. The use of bucketing and ironing in SATYA results in loss of some channel sharing opportunities. On the other hand, its running time scales exponentially with the number of neighbors with whom channel sharing is considered.

In both SMALL~\cite{small} and \cite{Hoefer-comm12}, the auctioneer collects sealed bids from all the bidders and uses VCG pricing to achieve a truthful auction. VCG pricing is considered truthful only if the underlying allocation scheme is optimal but finding an optimal channel assignment is NP hard. SMALL relies on obtaining the optimal solution for graph coloring underneath. \cite{Hoefer-comm12}, on the other hand, approximation algorithms are converted to truthful mechanisms using the randomized meta-rounding technique and is used to achieve truthfulness in expectation; however as noted before applying randomized meta-rounding results in prohibitive running times for large conflict graphs.

\subsubsection{Proposed Auction Design}
We develop an iterative auction mechanism that is fundamentally different from sealed bid auction schemes discussed above. Specifically, we seek an iterative auction that has the same outcome as a Vickrey auction. Moreover, as stated at the outset, we would like our auction mechanism to support marginal valuations, heterogeneous spectrum availability and channel sharing.

Iterative auctions have multiple advantages over sealed bid auctions. The foremost advantage of an iterative auction is that it provides a simpler means to achieve truthfulness. Another compelling advantage is that iterative auctions are transparent in the way they determine the outcome of an auction. In other words, bidders can verify and validate the auction outcome. To see why this might be important, notice that when bidders do not pay the bidding price, they can doubt the correctness of the auction scheme. In fact, a frequently mentioned problem with sealed bid auction schemes is that the auctioneer could create a fake second-highest bid after receiving all the sealed bids from the bidders in order to increase its revenue. This is a non-issue with iterative auctions as faking a second highest bid might result in revenue loss to the auctioneer as it is unaware of the private values of the highest bidder.

Yet another advantage of iterative auctions is that they are better at protecting the privacy of bidders' valuations. It is preferable for bidders not to be required to disclose their private values as they could be based on sensitive information. When compared to sealed bid (non-iterative) auctions, iterative auctions incur lower information revelation as they do not share their value with the auctioneer. While several solutions are available to protect the privacy of bidders in a sealed bid auction, they either require a third party to compute the outcome of the auction or require cryptography techniques. Also, the desire to realize transparent auction (i.e., the ability for bidders to validate the auction outcome) might compromise privacy with sealed bid auctions by requiring disclosure of several other bids.




The basis for our proposal, \texttt{VERUM}, is the ascending-bid multi-unit auction proposed in \cite{ausubel2004} that is simpler yet shown to be efficient and also replicate the outcome of Vickrey auction. But Ausubel's mechanism~\cite{ausubel2004} was not intended for dynamic spectrum sharing. As such it does not account for any of the unique characteristics associated with the dynamic spectrum allocation in general and TVWS spectrum secondary sharing in particular. Spatial reuse, which allows multiple users to be allocated the same channel provided they do not interfere with each other, is one such characteristic that is already mentioned above. For example, consider the conflict graph shown in Fig.~\ref{FigExample1} (a) . If HWSN E is using channel 1, it is still available for use by HWSNs A, B and C. Heterogeneous spectrum availability discussed in section 3.1.3 is another characteristic. Channel sharing discussed in section \ref{sharingsection} is yet another characteristic that needs to be supported. Our contribution in terms of auction design lies in adapting the mechanism in \cite{ausubel2004} to factor in all the above mentioned characteristics while preserving its desirable properties like simplicity, truthfulness and efficiency. To the best of our knowledge, this is the first time Ausubel's mechanism has been applied in the dynamic spectrum management context even though it has been around for some time.

We now describe \texttt{VERUM}. Referring to Figs. \ref{auctionmodel}(a) and \ref{auctionmodel}(b) , the auctioneer (HWSN spectrum manager) announces at the beginning of the epoch a reserve price for the channels and bidders (HWSNs) respond with the number of channels they are willing to buy at price $p_1$. The round price controls the demand from each bidder in the sense that the number of channels it can bid is determined by the number of channels within its private marginal valuations that have higher valuations than the current round price. In the example shown in Fig.~\ref{FigExample1}(a), if the current round price is 13, then demand from node B is 1 as it has only one channel that has a higher valuation than 13.

At each round $t$ with price $p_t$, the auctioneer determines if for any bidder $i$ the aggregate demand of bidder $i$'s neighbors in the conflict graph $D_{-i}(p_t) =\sum_{j \in N_i} D_j(p_t)$ is less than $x_i$, the number of channels available at $i$. If so, the difference is deemed clinched and the new channels clinched in this round are considered won by the bidder $i$ at that round price $p_t$. Note that to allow for spatial reuse, we view only the neighbors of a node in the conflict graph as its competing bidders. For example in the conflict graph shown in Fig.~\ref{FigExample1}(a): A competes with B and C in the auction; C competes with A, B and D; and E competes only with D. This is unlike the classical multi-unit Vickrey auction or Ausubel's mechanism where all bidders compete with each other.

To handle heterogeneous spectrum availability, we introduce the notion of exclusive channels. A channel $k$ is considered exclusive to HSWN $i$ if it is available for use only by $i$ and not by any of its neighbors. Some channels maybe exclusive to a HWSN right at the first round of the auction. Alternatively, channels may become exclusive to a HWSN in subsequent rounds (associated with higher reserve prices) when the demand of any of its neighbors reduces to zero. In each round $t$, we identify the set of exclusive channels and consider them clinched by the respective HWSNs at the round price $p_t$.


The above process repeats with increasing round prices until there is no demand from the bidders.

Fig.~\ref{FigExample1}(b) illustrates the working of the \texttt{VERUM} auction mechanism for the example in Fig.~\ref{FigExample1}(a). Note that in Fig.~\ref{FigExample1}(a), the set of channels available at each HWSN with ``A:''. As a specific example, HWSN A has two channels available (1 and 2). The number of channels a bidder is assigned is limited by the number of channels it has available. In the example, A can be assigned at most 2 channels even if its demand is more.

  \begin{figure*}
  \begin{center}
  \begin{tabular}{cc}
   \subfigure[]
{
	\includegraphics[width=2.4in]{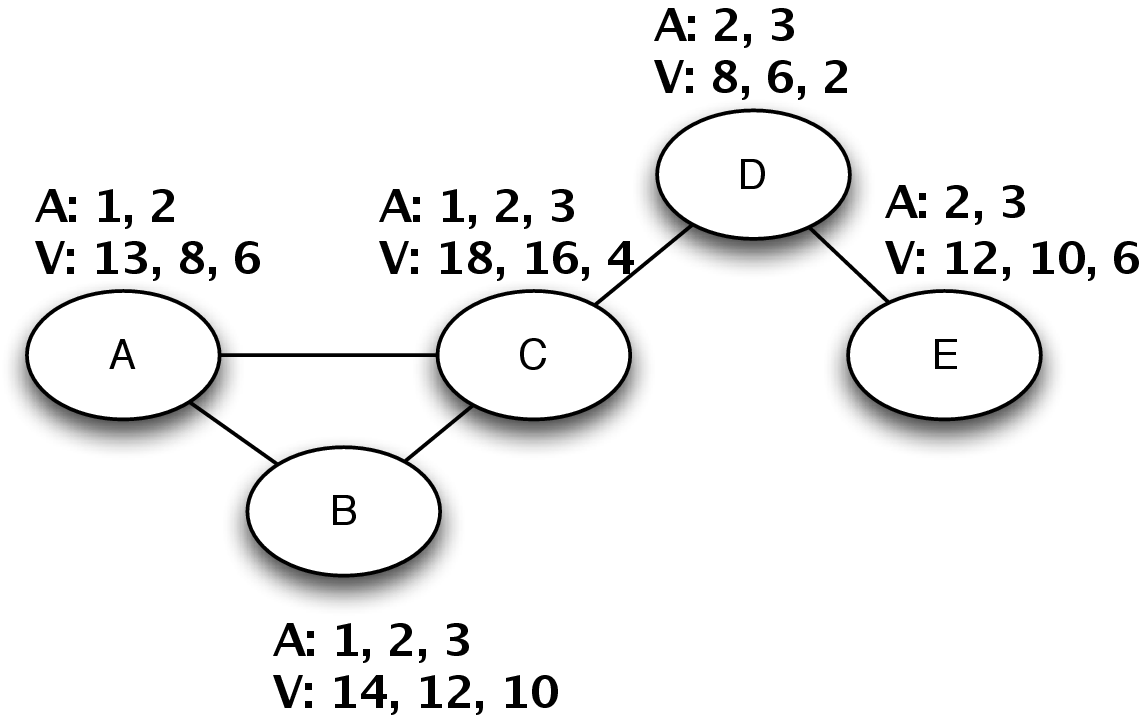}
 }
 &
   \subfigure[]
{
	\includegraphics[width=2.4in]{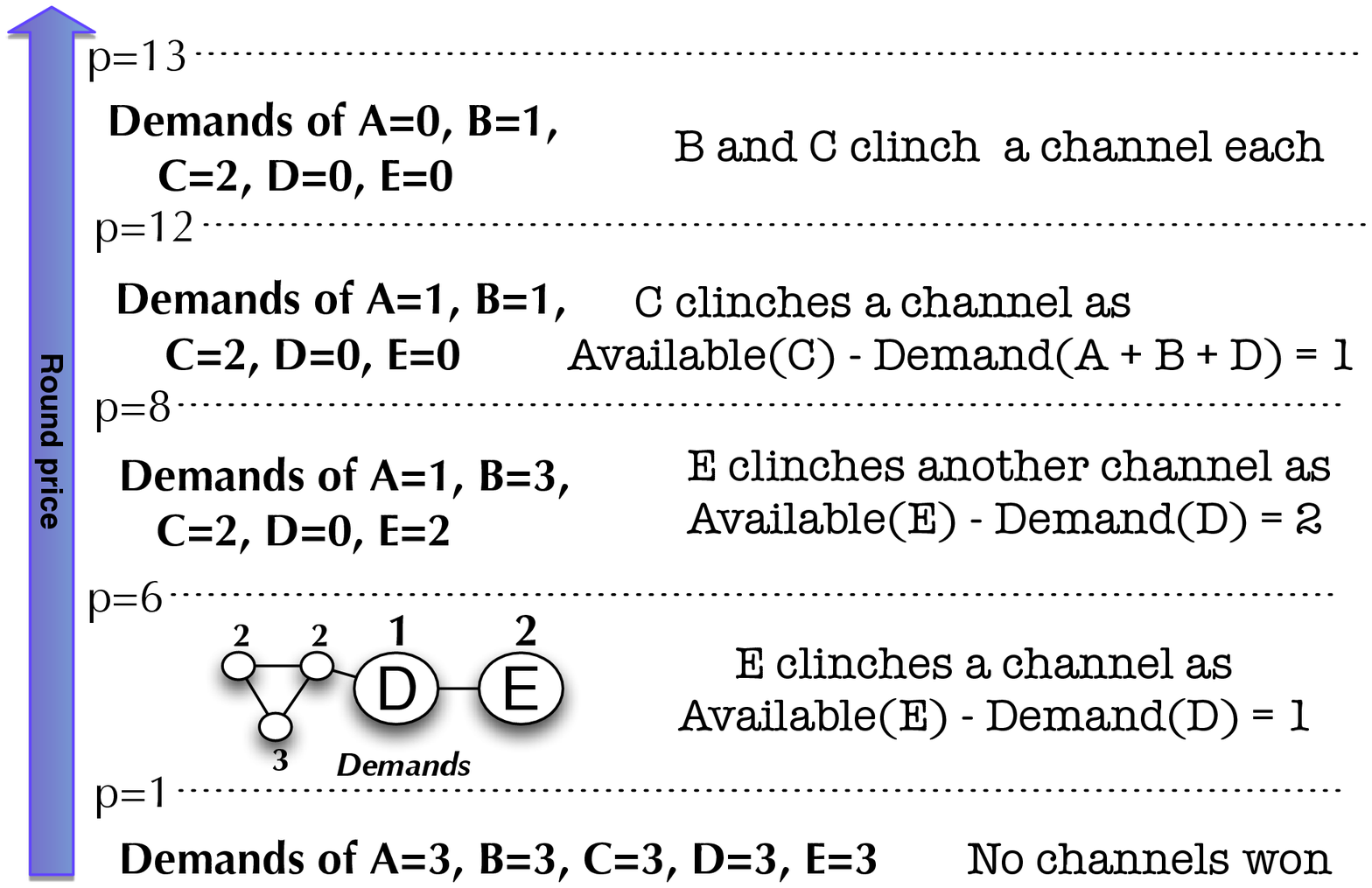}
 }
 \end{tabular}
\end{center}
 \caption{(a) Example conflict graph with 5 HWSNs A, B, C, D and E (b) Example illustrating the working of \texttt{VERUM} auction mechanism }
  \label{FigExample1}
\end{figure*}

%

In the first round of the auction with price $p=1$, the five bidders A, B, C, D, and E bid for 3 channels each (as they have higher valuations than the current round price). Since there is excess demand, the auction proceeds to subsequent rounds with the price getting incremented at each round. At price $p=6$, the cumulative demand of bidder E's neighbors is 1, where as 2 channels are available for use by E hence it is assured of winning at least one channel. Considering this, a channel is deemed clinched by bidder E at price $p=6$. Similarly at price $p=8$, as bidder D's demand drops to 0, E clinches another channel. At this point bidder E has clinched two channels as $x_E - D_{-E} = 2$, bidder E is assured of at least two channels. At price $p=12$, the cumulative demand of bidder C's neighbors (A, B and D) is 2 whereas there are 3 channels available for use by C. Hence bidder C clinches a channel at price $p=12$. Finally,  at price $p=13$, bidder A's demand drops to 0 and bidders B and C win a channel each. Since no more channels can be assigned at this point, the auction comes to an end with bidders B, C and E winning 1, 2, and 2 channels, respectively.



It can be clearly seen from the above example that the result of the auction is efficient: the auction has allocated the channels to the bidders who value them the most. The formal proof is provided in section~\ref{proofs}. It can also be seen that the resultant pricing for channels won is equivalent to that of multi-unit Vickery auction. For example, bidder C wins its first channel at the second highest losing bid ($p=12$) among its neighbors in the conflict graph and the second channel at the highest losing bid ($p=13$). Similarly, bidder B wins a channel at the highest losing bid amongst its neighbors.

Although the mechanism computes winner determination and payments effectively, it does not allocate channels. In the example discussed above, although we determine that the bidders B, C, and E won 1, 2, and 2 channels respectively, we did not identify the specific channels they won. This is because we consider all items (channels) as substitutes in the auction mechanism. We determine the actual channel allocation using a greedy algorithm. We select the channel that is available in the least number of the winning node's neighbors. This directly reduces the number of HWSNs for whom a channel becomes unavailable.


From the above description, we can determine that \texttt{VERUM} has a time complexity of $O(LCn)$ where $L$ is the number of rounds that the auction takes and $n$ is the number of HWSNs in the network and $C$ is the total number of channels available in the network. This shows that \texttt{VERUM} scales linearly with the number of HWSNs in the network.

\subsection{Channel Sharing Case}
\label{sharingsection}

Here we show how to extend \texttt{VERUM} to support channel sharing among interfering HWSNs. Towards this end, we introduce the notion of {\em usable channel opportunities}. The usable channel opportunities represent the total number of opportunities potentially available to each HWSN to use the available channels. In the exclusive use case, each HWSN has no more than one opportunity to use an available channel among a set of conflicting HWSNs. However with channel sharing enabled, conflicting HWSNs may have more than one opportunity between them to use a channel.

The usable channel opportunities of a HWSN $i$, $C_i^{opp}$, is defined as:

\begin{equation}
  C_i^{opp} = \sum_{k=1}^{C} \sum_{j\in N_i} F_j(k) \times X_i(k)
\end{equation}

where $F_j(k) \in \{0,1\}$ is the channel usability factor that indicates if channel $k$ can be used by HWSN $j$.

The mechanism to determine if a HWSN can use a channel can be as simple as fixing the number of HWSNs per channel, or based on a more realistic function like we do in the following. Specifically, we use a function of interference temperature and bandwidth given by:

 \begin{equation}
  F_p(k) = \left\{ \begin{array}{ll}
         1 & \mbox{if $T(p,b_p(k)) < \tau$ and $\sum_{q\in N_p} b_q(k) < 1-b_p(k)$};\\
         0 & \mbox{otherwise}.\end{array} \right.
 \end{equation}
where $T(p,b_p(k))$ is the interference temperature, $b_p(k)$ is the fraction of the channel bandwidth that the bidder $p$ intends to use on channel $k$ and $\tau$ is the threshold representing tolerable interference level. The interference temperature is a metric proposed in the literature~\cite{intf-temp} to quantify the interference among dynamic spectrum users. It is similar to noise temperature and used to measure the power and bandwidth occupied by interference.

Interference temperature can be computed for each HWSN $p$ as:

\begin{equation}
  T(p,b_p(k)) = \frac{Power_p (p,b_p(k)*B)}{kB}
\end{equation}
where $Power_p (p,B)$ is the interference power in watts centered at $p$ while using channel bandwidth $B$ in Hz and $k$ is the \textit{Boltzmann}'s constant.

With usable channel opportunities $C_i^{opp}$ defined as above, extending \texttt{VERUM} to allow channel sharing reduces to changing the criteria for clinching channels. Specifically, a bidder $i$ clinches $m$ channels if its number of usable channel opportunities $C_i^{opp}$ exceeds the aggregate demand of its neighbors by $m$. The price paid by the winning bidder for a channel $k$ would now be a fraction of the exclusive price: $Price_i^{shared}(k) = Price_i(k) * b_i(k)$, where $Price_i(k)$ is the price paid if $k$ was an exclusive channel and $b_i(k)$ is the fraction of channel $k$ that HWSN $i$ utilizes.

Note that channel sharing realized as stated above would always result in spectrum utilization and winning bidders that is no worse than the exclusive use case. This is because even if the channel demands of conflicting HWSNs are high (e.g., $b_i (k) \approx 1$), the iteratively increasing price with each new round would drive out excess demand eventually so that Eq. 4 would be satisfied for at least one HWSN.

\subsection{VERUM Truthfulness and Efficiency}
\label{proofs}

\newtheorem{theorem1}{Theorem}
\begin{theorem}
\texttt{VERUM} is truthful.
\end{theorem}

\begin{proof}
In order to prove that an auction mechanism is truthful, we need to show: (i) the pricing function does not depend on the bid of the winning bidder; and (ii) it is monotonic, i.e., if bidder $i$ wins a channel at bid $p$ then he will win the channel at any bid $p^* > p$.

It is indeed the case that pricing function in \texttt{VERUM} does not depend on the bid of the winning bidder. In any given round, the number of channels won by a bidder $i$ is not dependant on $i$'s demand but instead on the cumulative demand of $i$'s conflicting neighbors. Even more crucially, the price that $i$ needs to pay for the channels it clinches in a round $t$ is the round price $p_t$, which does not have any relation with $i$'s bid.



Now to the monotonicity. Assume HWSN $i$ won a channel at bid $p$ and at any of its subsequent bid $p^* > p $, the cumulative demand of $i$'s neighbours $D_{-i}(p) >= D_{-i}(p^*)$.  The only way $i$ could not win the channel at higher price $p^*$ is if the aggregate demand of $i$'s neighbors increases with the bid $p^*$. This is not possible since we have assumed that the marginal valuations are weakly decreasing (see Section 3.1.4), which results in monotonically non-increasing demands with each new round. Thus $i$ will always win the channel at any bid $p^* > p$.
\end{proof}


\begin{theorem}
\texttt{VERUM} yields an efficient allocation, i.e., spectrum is allocated to bidders who value it the most.
\end{theorem}

\begin{proof}
The proof is by contradiction. Suppose that bidder $i$ wins channel $k$ with valuation given by $V_i(k) < max_{j \in N_i} V_j(k)$ at price $p_t$. Note that for the purposes of this proof we consider that each channel is distinct, a general case compared to what we had so far viewing all channels as identical and as substitutes for each other.

When the bidder $i$ wins channel $k$ at round $t$ with round price $p_t$, as per \texttt{VERUM}, the aggregate demand for channel $k$ from $i$'s conflicting neighbors must be zero. That would happen if and only if $\forall j \in N_i V_j(k) \le p_t$ and $V_i(k) > p_t$. In other words,  $D_{-i}(p_t)(k) = 0$ and $V_i(k) > max_{j \in N_i} (V_j(k))$, a contradiction.

Hence we have the following. If a bidder $i$ wins a channel $k$ then $ \forall j \in N_i ~~~ V_i(k) >  V_j(k)$ --- the channel is allocated to its highest valuing bidder, proving that the auction mechanism is efficient.
\end{proof}


\subsection{Revenue Maximizing Integer Linear Programs}
\label{ilpappendix}

Even though revenue maximization is not the main objective behind the design of \texttt{VERUM}, we are still interested in evaluating the revenue generated using it. In order to benchmark its performance in terms of revenue, below we formulate the revenue maximizing spectrum allocation problems corresponding to the exclusive use and channel sharing models as integer linear programs.

We first introduce the notion of cumulative clinches ${C_i^t}$ that represents the total number of channels won by a bidder $i$ including the number of exclusive channels at $i$ until round $t$. Using this notion, the number of channels won until any given round $t$ can be quantified. Specifically it is defined as:

 \begin{equation}
 C_i^t = \max \left\{0, x_i - D_{-i}(p_t) + E_i(p_t)
\} \right\}
\label{clinches}
 \end{equation}

where $x_i$ is the number of available channels at bidder $i$. $D_{-i}(p_t)$ is the cumulative demand of bidder $i$'s neighbors and $E_i(p_t)$ is the number of channels available exclusively at $i$, both till round $t$.

The number of channels clinched at round $t$,  ${c_i^t}$, can then be defined as:

\begin{equation}
c_i^t = C_i^t - C_i^{t-1}.
\end{equation}

Now the total price $S_i$ paid by bidder $i$ can be stated as:

\begin{equation}
 S_i = \sum_{t=1}^L p_t \times c_i^t
\end{equation}
where $L$ is the number of rounds the auction takes and $p_t$ is the price at round $t$.

%

\subsubsection{Exclusive Use Case}

If there are $n$ HWSNs participating in the auction, then linear program for the revenue maximization problem in the exclusive use case is as follows.

\begin{align*}
&Maximise : R = \sum_{i=1}^n S_i \;\;\; \mathtt{subject \;to} \\
&\forall i \sum_{\substack{k=1\\i\not=j}}^{C} \{Y_i(k) \times Y_j(k) \} = 0, \forall j \in N_i\\
&\forall i \sum_{k=1}^{C} \{Y_i(k) \times X_i(k) \}= y_i \\
\end{align*}

The first constraint ensures that the same channel is not assigned to conflicting neighbors and the second constraint ensures that only available channels are allocated to the HWSNs.

\subsubsection{Channel Sharing Case}

Similar to the exclusive usage model, we can formulate the revenue maximizing spectrum allocation problem as the following integer linear program.


\begin{align*}
& Maximise : R = \sum_{i=1}^n S_i \;\;\; \mathtt{subject \;to} \\
&\forall i ~~~ \forall k \in \{1, .., C\} ~~~ T(i,b_i(k)) \times Y_i(k)< \tau\;\;\; \\
&\forall i ~~~\forall k \in \{1, .., C\} \sum_{j\in N_i} \{b_j(k) \times Y_j(k)\} < 1-[b_i(k) \times Y_i(k)] \\
&\forall i \sum_{k=1}^{C} \{Y_i(k) \times X_i(k) \}= y_i \\
\end{align*}

While sharing the channel among multiple HWSNs, the first constraint ensures that the interference is less than a given threshold. The second constraint ensures that enough bandwidth is available for all the HWSNs sharing a channel. The last constraint is as in the exclusive case.

\section{Evaluation}
\label{eval}

For our evaluation, we follow the auctioning based coordinated TVWS spectrum sharing framework described in section~\ref{model} with the HWSN Spectrum Manager as the auctioneer and participating HWSNs with non-zero demand as bidders. We compare the results obtained with different auction mechanisms. Specifically, we compare \texttt{VERUM} with channel sharing enabled against VERITAS~\cite{veritas} and SATYA~\cite{satya}, the two existing truthful and efficient auction schemes that can be applied for our problem setting. Recall that VERITAS does not support channel sharing, whereas SATYA supports channel sharing as well as heterogeneous channel availability. For SATYA, we use a bid price based bucketing function as was done by its authors in \cite{satya}. To benchmark the above three mechanisms with respect to the optimum, we use the ILP from \ref{ilpappendix} with Vickrey pricing. Results shown for optimal solution are obtained by solving the ILP using the GUROBI solver\footnote{\url{http://www.gurobi.com/}}.

To examine the impact of the density of HWSNs, we consider two different real residential environments: dense-urban and urban. The house and building layout data for each of these environments represent 1 square kilometer areas. These data are obtained from the UK Geographical Information System (GIS) database~\cite{maps-uk} for parts of London. The dense-urban residential environment comprises of 5456 houses or buildings, whereas 2435 houses/buildings are present in the urban environment. The number of TVWS channels available at any given location is computed based on pixelized data of coverage maps for UK TV transmitters. Across the two environments we consider, we find that the maximum number of TVWS channels available over all locations is 21 but some locations may have fewer than maximum channels in the case of heterogeneous channel availability.

 Fig.~\ref{uncoordinated} shows average number of potential interferers for the dense-urban and urban scenarios considered. As the interferers are computed purely based on node (house) distribution and interference range without any assumption on spectrum coordination, Fig.~\ref{uncoordinated} reflects the worst-case interference experienced with uncoordinated TVWS spectrum sharing among HWSNs. It can be seen that the potentially number of interferers are significantly high, especially for the dense-urban scenario and higher interference ranges. These results support the coordinated spectrum sharing approach advocated in this paper.

\begin{figure}[h]
  \begin{center}
   \includegraphics[width =2in]{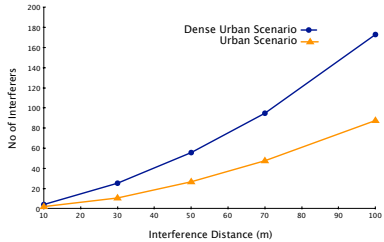}
    \caption{Average number of potential interferers in urban and dense-urban scenarios with different interference ranges.}
\label{uncoordinated}
 \end{center}
\end{figure}

We realize diverse spectrum demands from different HWSNs as follows. The average demand across all subscribed and active HWSNs is a variable input parameter specified as a percentage value. Demand for each HWSN (again as a percentage value) is then randomly generated so that overall average conforms to the specified input value. The demand for a HWSN in terms of number of channels needed is obtained by taking the product of its percentage demand and available number of channels (e.g., A HWSN with 80\% demand and 5 available channels needs 4 channels). Unless otherwise specified, the default value for average demand across all HWSNs is set to 60\%. Marginal valuations for the set of channels requested by a HWSN are also randomly generated so that they are weakly decreasing and each value falls within the range of [0, 100] virtual currency units.

For evaluation purposes, we use a fixed interference range based model to study the impact of different interference conditions. By default we set the interference range to 30 meters. We also analyse the performance of VERUM with different interference ranges in section \ref{additionalResults}.

\subsection{Metrics}

We use the following metrics in our evaluations:

\begin{itemize}

\item {\em Revenue}: The revenue the auctioneer (HWSN spectrum manager) obtains for managing access to TVWS spectrum in an interference-aware manner. This is the sum of the market clearing prices paid by all $k$ winning bidders in the auction $R = \sum_{i=1}^k S_i$ where $S_i$ is given by Eq. 8. Note that the primary goal for \texttt{VERUM} is to maximize the utility for all auction participants rather than maximizing revenue, so we set the initial round price (reserve price) for \texttt{VERUM} to a small value (10 virtual currency units by default in our evaluations) and increase the price with the sole aim of removing excess demand.

\item {\em Spectrum Utilization}: The percentage of available channels at each HWSN that are allocated (to any HWSN), averaged across all HWSNs. One of the primary reasons for the interest around TVWS spectrum is to achieve better spectrum utilization. An efficient mechanism in the spectrum utilization sense should be able to allocate all the available channels given enough demand for the channels.

\item {\em Percentage of Winners}: The percentage of bidders who are allocated at least one channel at the end of the auction. A higher value of this metric is preferred as that indicates users have incentive to successfully participate in the auctioning based coordinated access method. The number of winners, however, is constrained by the density of the area and the mutual interference among HWSNs in that area.

\end{itemize}

\subsection{Revenue}
\begin{figure}
  \begin{center}
  \begin{tabular}{c}
   \subfigure[]
{
	\includegraphics[width=2.4in]{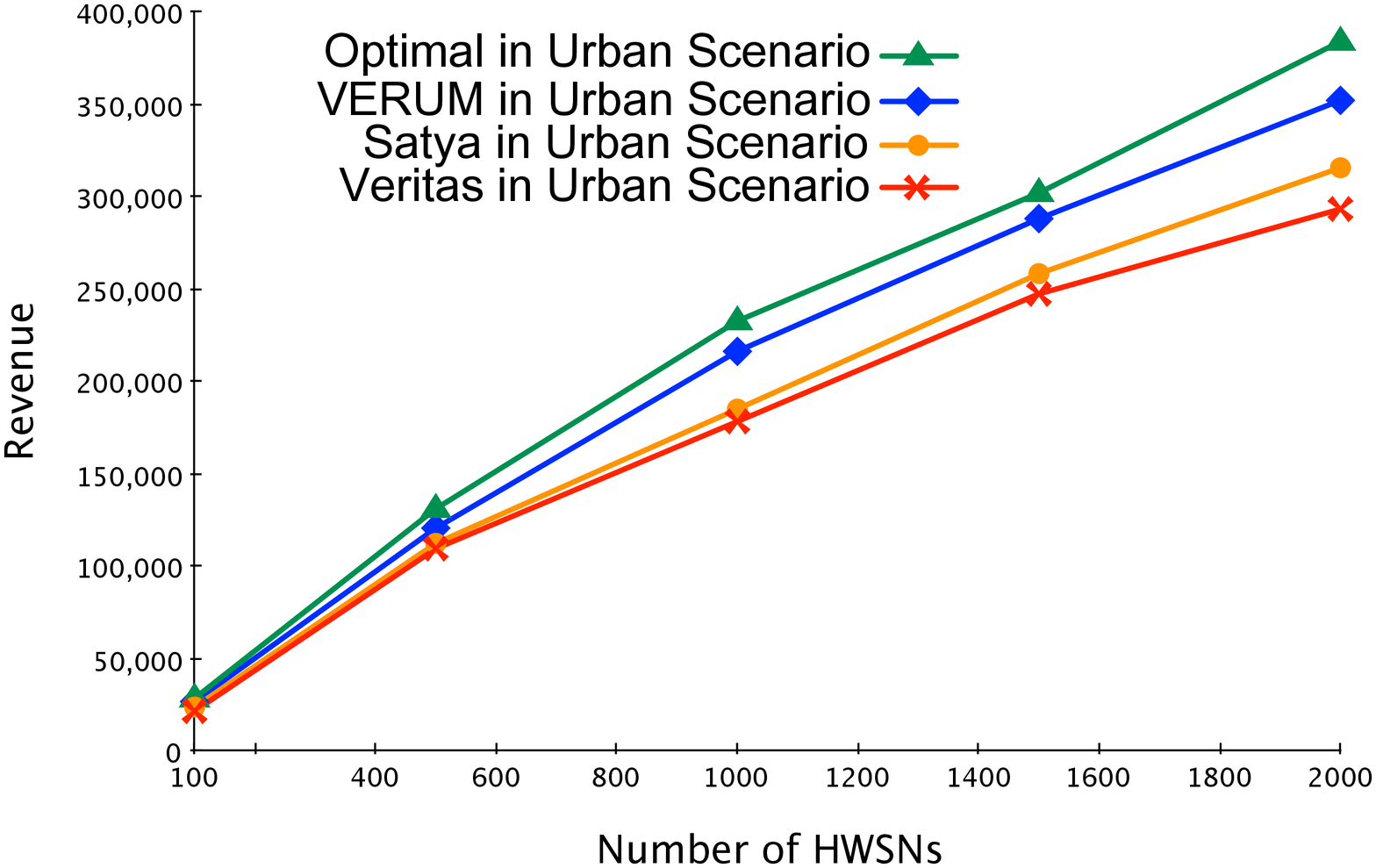}
        	\label{revenue_noofnodes_absolute}
 }
\\
   \subfigure[]
{
	\includegraphics[width=2.4in]{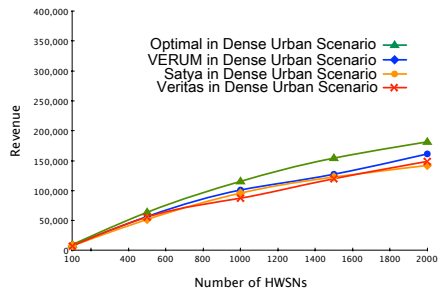}
        	\label{revenue_noofnodes_absolute_hd}
 }
 \\
   \subfigure[]
{
	\includegraphics[width=2.2in]{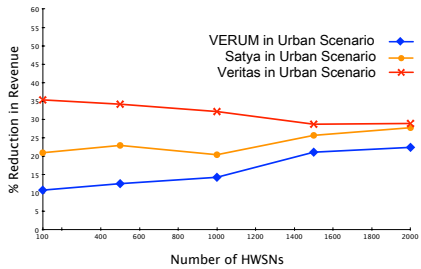}
        	\label{revenue_noofnodes}
 }
\\
   \subfigure[]
{
	\includegraphics[width=2.2in]{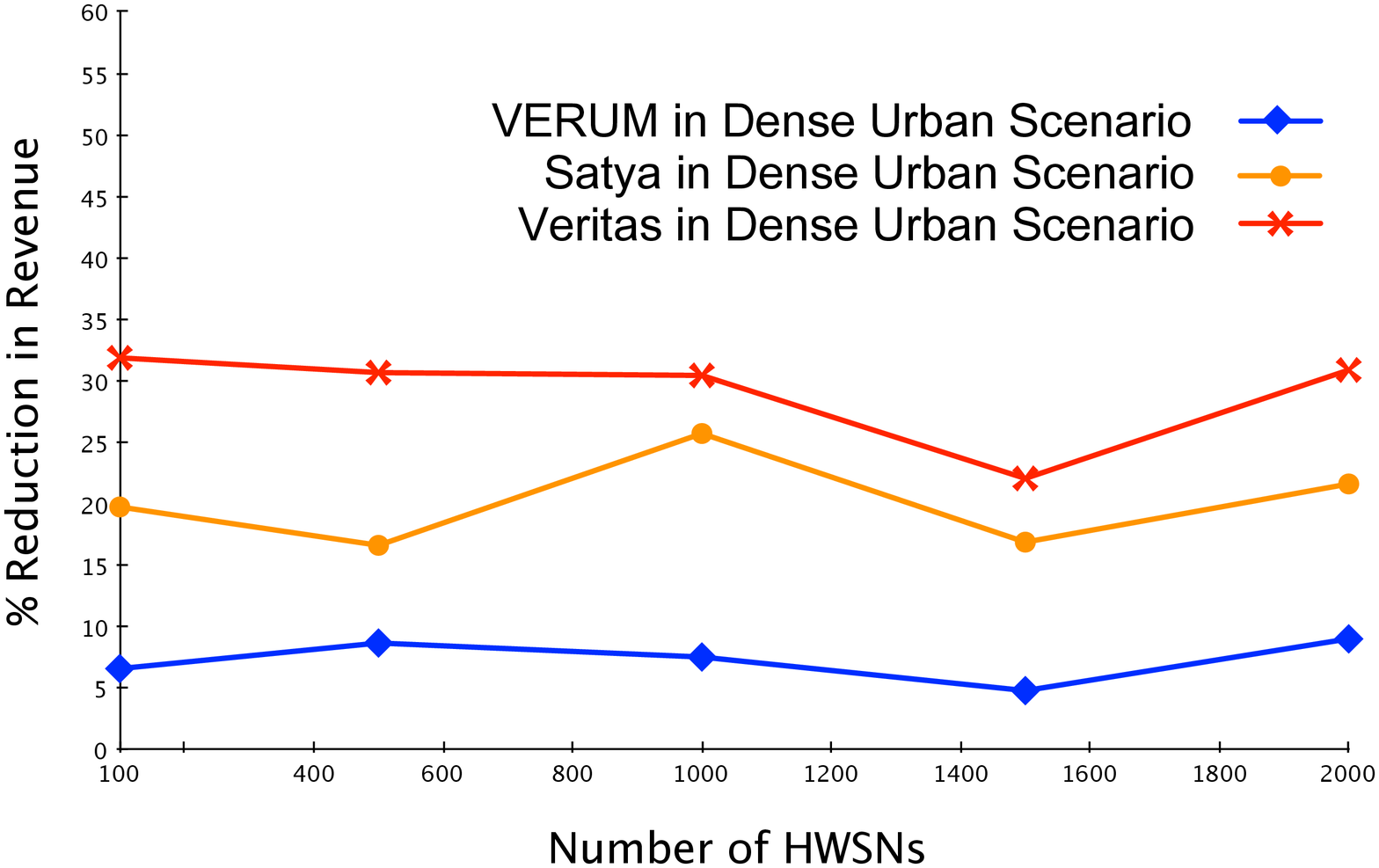}
        	\label{revenue_noofnodes_hd}
 }
 \end{tabular}
\end{center}
 \caption{Revenue with varying number of active subscribed HWSNs.}
  \label{revenue-nodes}
\end{figure}
 
Figs.~\ref{revenue-nodes} (a) and (b) show the revenue (in virtual currency units) obtained with different auction schemes with varying number of active subscribed HWSNs for urban and dense-urban scenarios respectively. Note that varying the number of HWSNs in either scenario only has the effect of varying the area under consideration, thereby allowing us to study how different mechanisms react to increasing scale in terms of area and number of HWSNs for a given density.  On the other hand, the effect of HWSN density can be observed by comparing urban and dense-urban scenarios. We notice that there is a substantial drop in revenue generated by the auction for all mechanisms in the dense-urban case when compared to the urban scenario. As the number of potential interferers is greater in the dense-urban scenario (see Fig.~\ref{uncoordinated}), there are correspondingly fewer channel reuse opportunities in that scenario. This in turn has the effect of fewer winners per channel and consequent reduction in revenue relative to the urban scenario.

\begin{figure}[h!]
  \begin{center}
  \begin{tabular}{c}
   \subfigure[]
{
	\includegraphics[width=2.4in]{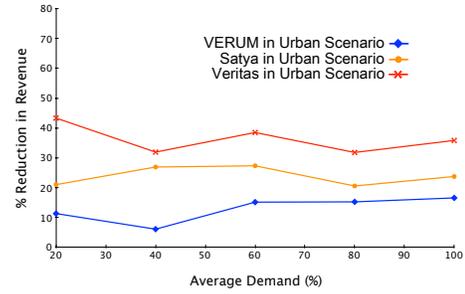}
        	\label{revenue_demand}
 }
 \\
   \subfigure[]
{
	\includegraphics[width=2.4in]{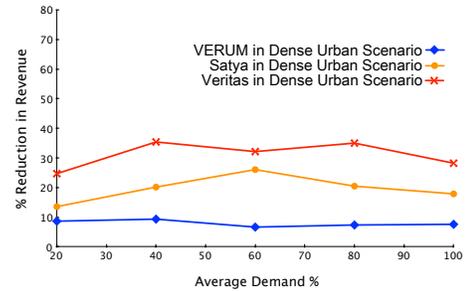}
        	\label{revenue_demand_hd}
 }

 \end{tabular}
\end{center}
 \caption{Revenue with varying demand.}
  \label{revenue-demand}
\end{figure}

To better highlight the benefit of \texttt{VERUM} for revenue generation compared to VERITAS and SATYA, we show percentage reduction in revenue for these three mechanisms with respect to the optimal case (ILP solution) in Figs.~\ref{revenue-nodes} (c) and (d) for urban and dense-urban scenarios. The non-monotonic trend seen for individual curves in those figures can be attributed to the fact that auction mechanisms and the optimal case have different objectives --- while the optimal seeks a revenue maximizing solution, the three auction mechanisms assign channels to the bidders with the highest valuations. This is also a key reason behind the reduced amount of revenue generated by \texttt{VERUM}, SATYA and VERITAS to different degrees compared to the optimal (non-zero percentage reduction in revenue in Figs.~\ref{revenue-nodes} (c) and (d)). The optimal case to maximize revenue may allocate channels to HWSNs with lower valuations if they allow for more winning bidders. On the other hand, auction mechanisms considered assign
channels to HWSNs who value it the most regardless of the consequence on revenue.

Comparing the three mechanisms in Figs.~\ref{revenue-nodes} (c) and (d), we see that VERITAS causes the most drop in revenue with respect to the optimal by as much as 35\%. This is because it does not support channel sharing. SATYA relatively fares better primarily due to its support for channel sharing. Even SATYA too results in close to 30\% reduction in revenue in some cases because the bucketing and ironing techniques it employs limit channel sharing opportunities and hence the revenue. Specifically, bidders in SATYA are not allowed to share the channel with a neighbor placed in a higher bucket. As \texttt{VERUM} does not impose such constraints to ensure truthfulness, it offers the best relative performance in all cases, mostly within around 10\% of the optimal and around 20\% revenue drop in the worst case. The harmful effect of constrained channel sharing support in SATYA is more noticeable in the dense-urban scenario (Fig.~\ref{revenue-nodes} (d)), where effective channel sharing is crucial to higher
revenue and number of winning bidders; in this scenario SATYA is more closer to VERITAS than \texttt{VERUM}.


Similar observations as above can be made even with diverse and varying demands. Fig.~\ref{revenue-demand} shows percentage reduction in revenue relative to optimal for the three mechanisms across a wide range of average demand settings from 20\%--100\% while keeping the number of active subscribed HWSNs fixed at 2000. Results in Fig.~\ref{revenue-demand} clearly demonstrate the effectiveness of \texttt{VERUM} over SATYA and VERITAS in optimizing revenue for the auctioneer while remaining truthful and efficient.


%

Fig.~\ref{revenue-demand} (c) and (d) demonstrates similar results as above for different interference ranges while keeping HWSNs fixed at 2000 and average demand at 60\%. In terms of reduction in revenue with respect to the optimal, \texttt{VERUM} offers more stable result across the two scenarios and different interference ranges. As increasing the interference range has a similar effect to increasing the density, the explanation given for comparison between the results in Figs.~\ref{revenue-nodes} (c) and (d) (and Figs.~\ref{revenue-demand} (a) and (b)) applies till a certain interference range, to a further extent for the relatively sparser urban scenario. But for very high interference ranges, conflict graph becomes fully connected severely limiting spectrum reuse and consequently the revenue for all three mechanisms.

\subsection{Spectrum Utilization and Percentage of Winners}

We now look at the relative performance of \texttt{VERUM}, SATYA and VERITAS in terms of the other two metrics: spectrum utilization and percentage of winners. Fig.~\ref{util-winners-demand} shows the results. It can be observed that spectrum utilization is higher in the dense-urban scenario as a result of greater amount of channel sharing among HWSNs. Percentage of winners is seen to be lower in the dense-urban scenario. As noted already in section 5.2, with higher number of potential interferers channel reuse opportunities reduce, which in turn reduces the number of winners despite channel sharing.
\begin{figure}[h!]
  \begin{center}
  \begin{tabular}{c}
   \subfigure[]
{
	\includegraphics[width=2.4in]{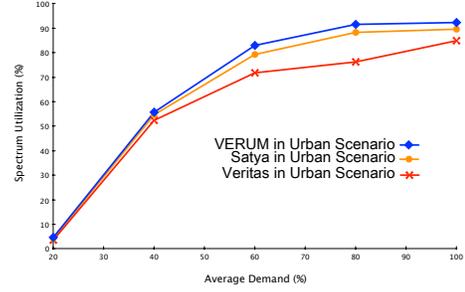}
 }
 \\
 
   \subfigure[]
{
	\includegraphics[width=2.4in]{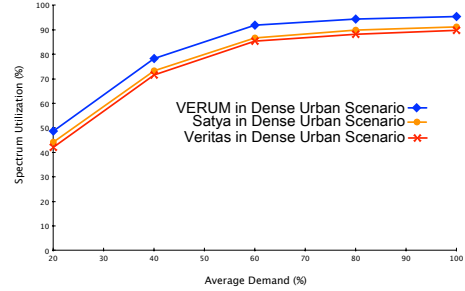}
 }
 \\
 \subfigure[]
{
	\includegraphics[width=2.4in]{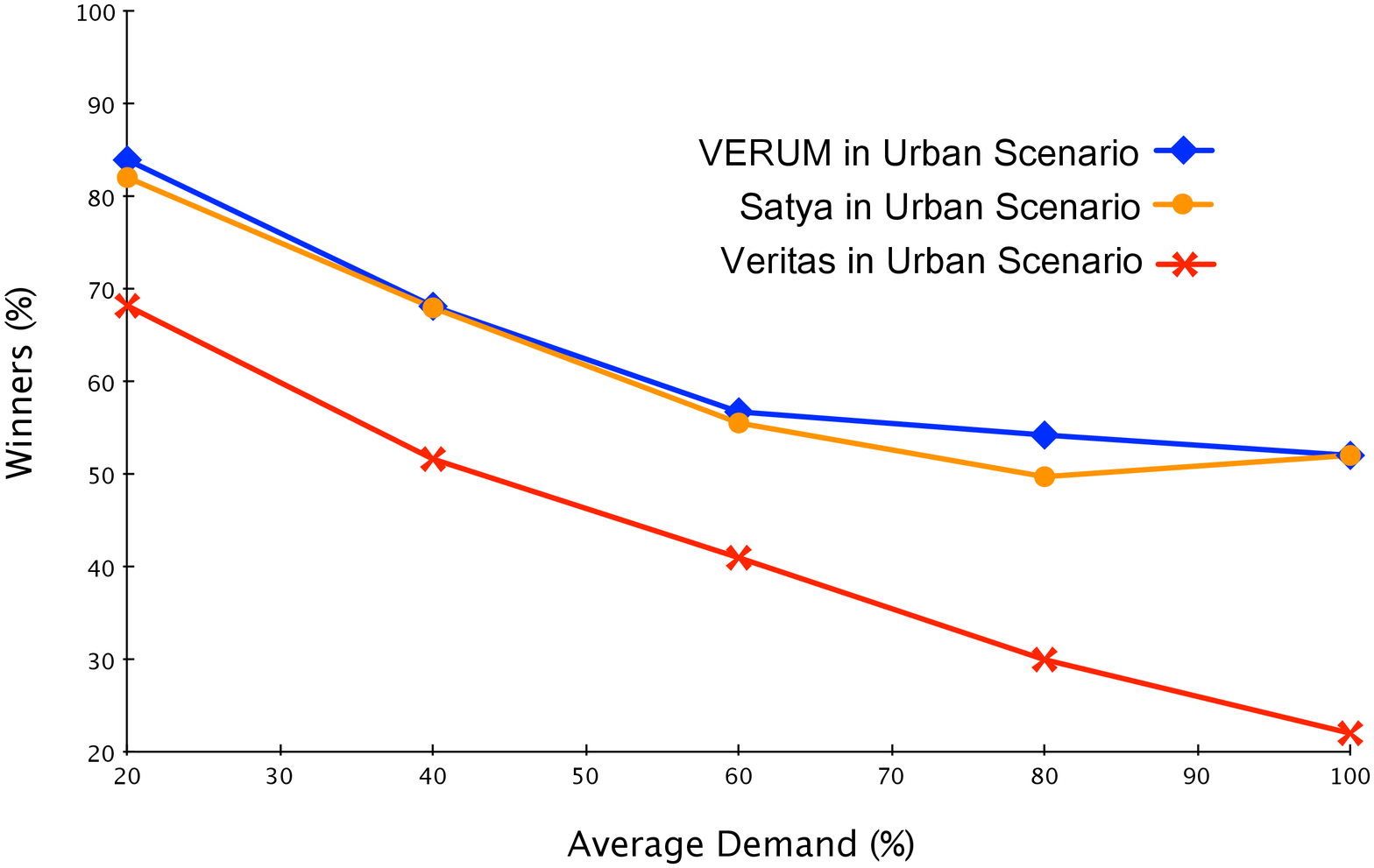}
 }
\\
   \subfigure[]
{
	\includegraphics[width=2.4in]{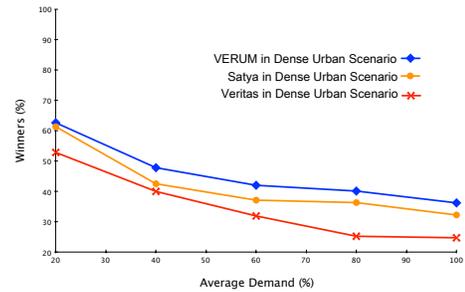}
 }
 \end{tabular}
\end{center}
 \caption{Spectrum utilization and percentage of winning bidders with varying demand.}
  \label{util-winners-demand}
\end{figure}
We can observe from Fig.~\ref{util-winners-demand} that \texttt{VERUM} does comparatively better both in terms of spectrum utilization and percentage of winners. VERITAS performs worse in all cases as it lacks support for channel sharing. SATYA gains in comparison with VERITAS as it allows channel sharing. But the bucketing approach underlying SATYA makes it lose out on some channel sharing opportunities, explaining the performance gap between SATYA and \texttt{VERUM}. This is more clearly seen in the dense-urban scenario where effective channel sharing is more crucial. In that scenario, SATYA's spectrum utilization is almost like that with VERITAS, and its differences with VERUM for both the metrics are also more appreciable.

\subsection{Impact of VERUM Parameters}

In this section, we look at the sensitivity of \texttt{VERUM} to two of its underlying parameter settings: reserve price and step size. To study the impact of these two parameters we focus on the case with 2000 HWSNs and 60\% average demand.

\subsubsection{Reserve Price}



Reserve price in \texttt{VERUM}  (the price announced by the auctioneer in the first round) could be 1 or start at a higher level $p>1$. A higher reserve price can improve the revenue as it is the minimum price that bidders need to pay to get access to a channel. But increasing the reserve price can also result in lower spectrum utilisation and in turn loss in revenue due to some channels remaining unallocated even when there is a real demand. The latter can occur because a high reserve price effectively throttles bidders' demand by allowing them to bid only for the number of channels in their demand for which they have a valuation greater than the reserve price.
 \begin{figure}[h!]
  \begin{center}
  \begin{tabular}{c}
   \subfigure[]
{
	\includegraphics[width=2.4in]{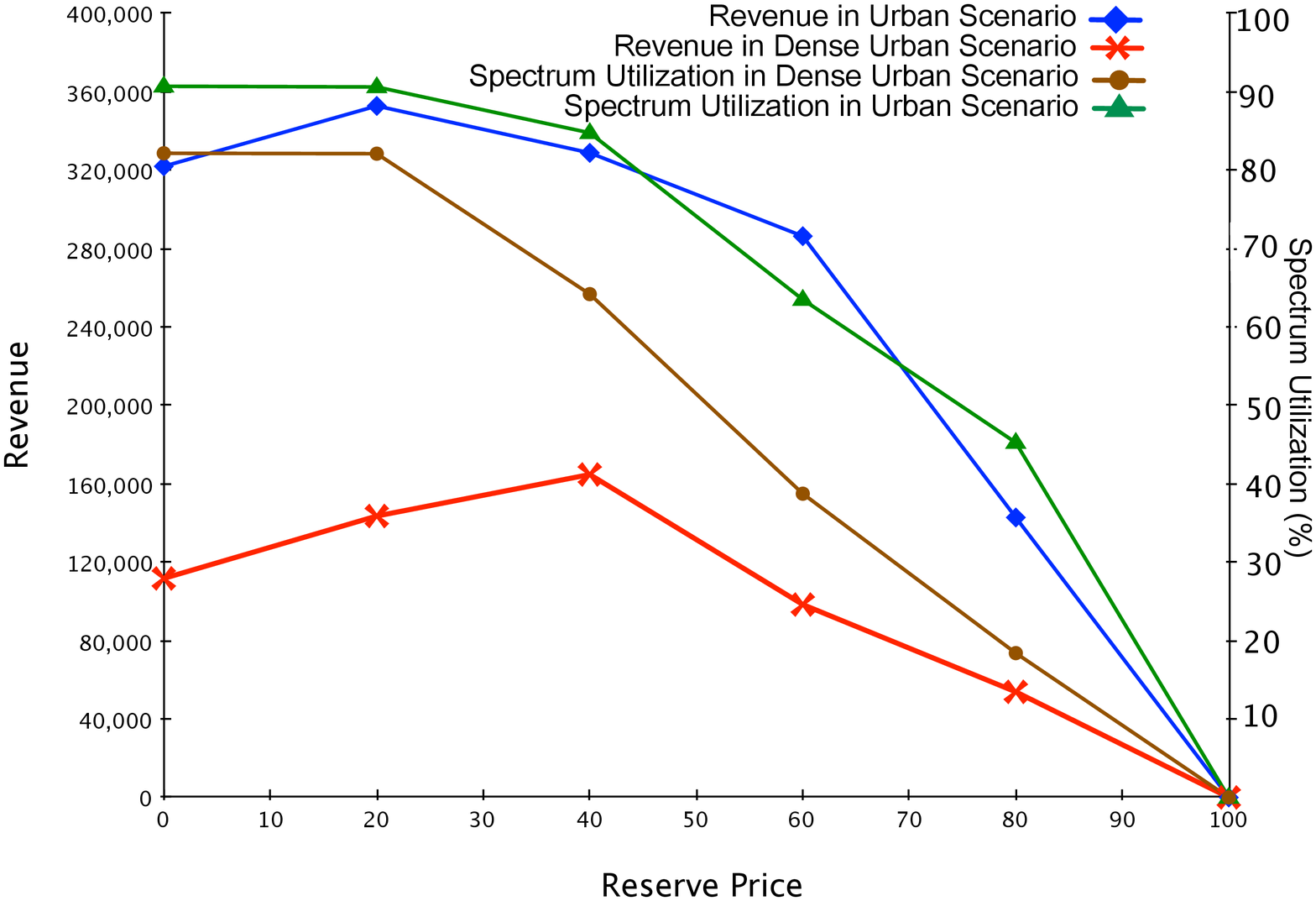}
 }
\\
   \subfigure[]
{
	\includegraphics[width=2.4in]{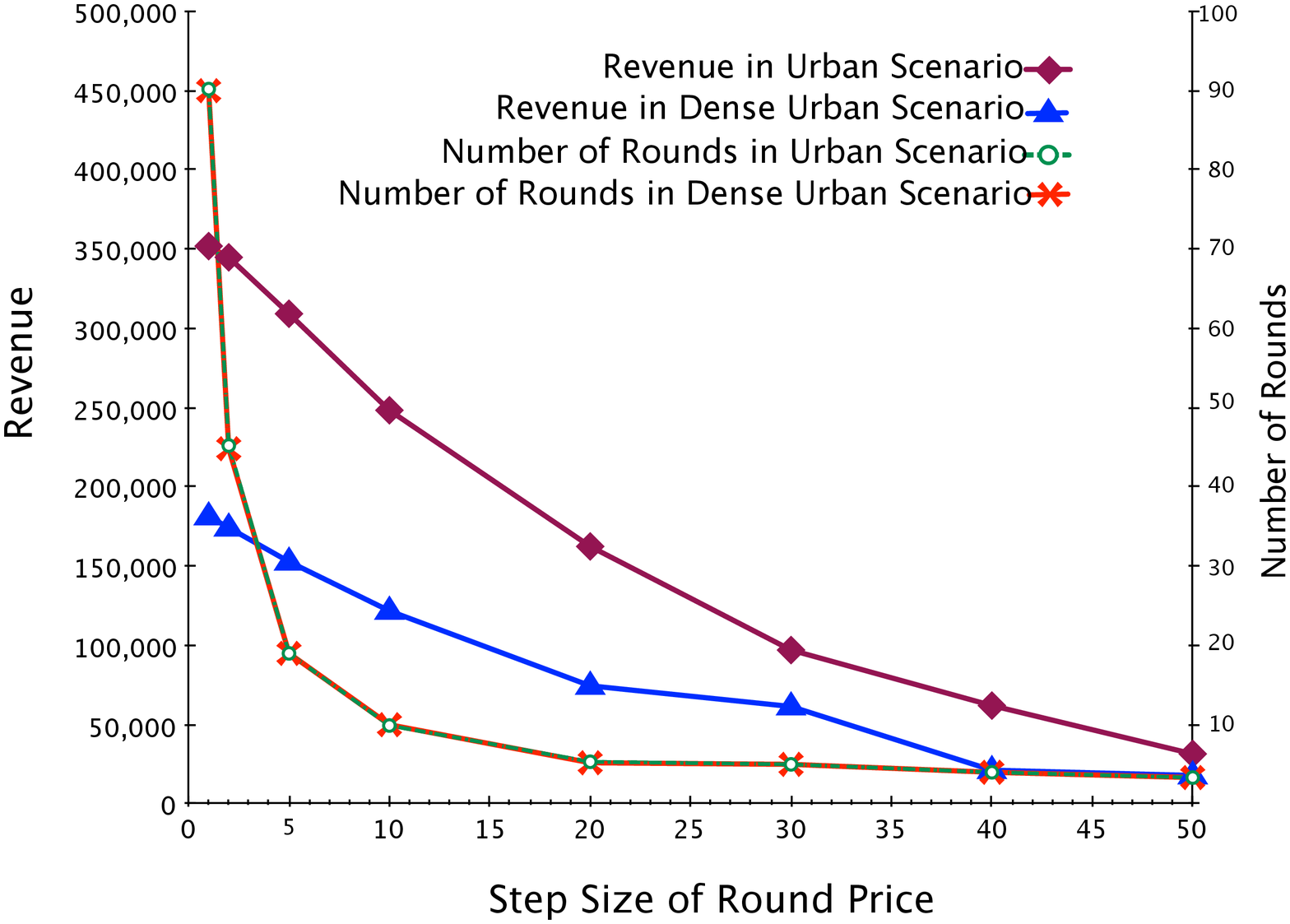}
 }
 \end{tabular}
\end{center}
 \caption{Impact of reserve price (a) and step size (b) on VERUM.}
  \label{reserve-step}
\end{figure}
The interaction between revenue and spectrum utilization due to reserve price is shown in Fig.~\ref{reserve-step} (a). The maximum spectrum utilization is achieved when the reserve price is $0$. The reduction in spectrum utilization at a higher reserve price represents the percentage of channels that remain unallocated even when there is demand. Unlike spectrum utilization, revenue increases with reserve price until it is $20$ for the urban scenario ($40$ for the dense-urban) as the effect of the reserve price being the minimum price paid for a channel is higher than the revenue lost due to unallocated channels. At higher reserve prices, the revenue lost due to unallocated channels is much higher than any gain from higher reserve prices, resulting in a much lower revenue.

\subsubsection{Step Size}

Step size is another parameter that influences the behavior of \texttt{VERUM}. It is the increment amount by which price is increased in one round to the next. Fig.~\ref{reserve-step} (b) shows the impact of different step sizes on revenue and duration of auction (in terms of number of rounds to completion). Note that the auction completes when none of the HWSNs have a valuation greater than the current round price; in other words, the demand for all of them effectively goes to zero. Both urban and dense-urban scenarios have identical number of rounds because we use the same marginal valuations for both the scenarios.

As expected, the highest revenue is obtained with the smallest step size. On the other hand, the auction takes most number of rounds to complete. Increasing the step size decreases the revenue because a higher step size drives out demand quickly and get the auction to a point when effectively all HWSNs have zero demand even though there are still unallocated channels, explaining the drop in revenue. Clearly, a higher step size leads to auction completing in fewer rounds. Interesting point to note is about how reduction in revenue with increasing step size relates to the reduction in number of rounds. It can be seen from Fig.~\ref{reserve-step} (b) that increasing the step size from 1 to 5 results in about 15\% loss of revenue but reduces the number of rounds by 80\%, suggesting a value for step size that keeps the duration and overhead of auction minimal without hurting the revenue much.

\subsection{Heterogeneous Spectrum Availability}

Spatio-temporal heterogeneity in spectrum availability is inherent to the TVWS spectrum, even more so when TVWS devices rely on different spectrum access models (e.g., coordinated, uncoordinated) or when TVWS devices are divided amongst multiple HWSN spectrum managers. To understand how different auction mechanisms behave in such a situation, we model heterogeneous spectrum availability as follows. Given a percentage of third party WSNs (40\% by default in our evaluations), we randomly select an equivalent number of homes in the scenario considered (urban or dense-urban) and label them as third party WSNs. These third party WSNs are assumed to consume a certain specified fraction of the available channels at their respective locations (0.7 by default in our evaluations). The combined effect of the above is that different active and subscribed HWSNs would see different sets of channels available, thus realizing heterogeneous spectrum availability situations.

 \begin{figure}[h]
  \begin{center}
  \begin{tabular}{c}
   \subfigure[]
{
	\includegraphics[width=2.in]{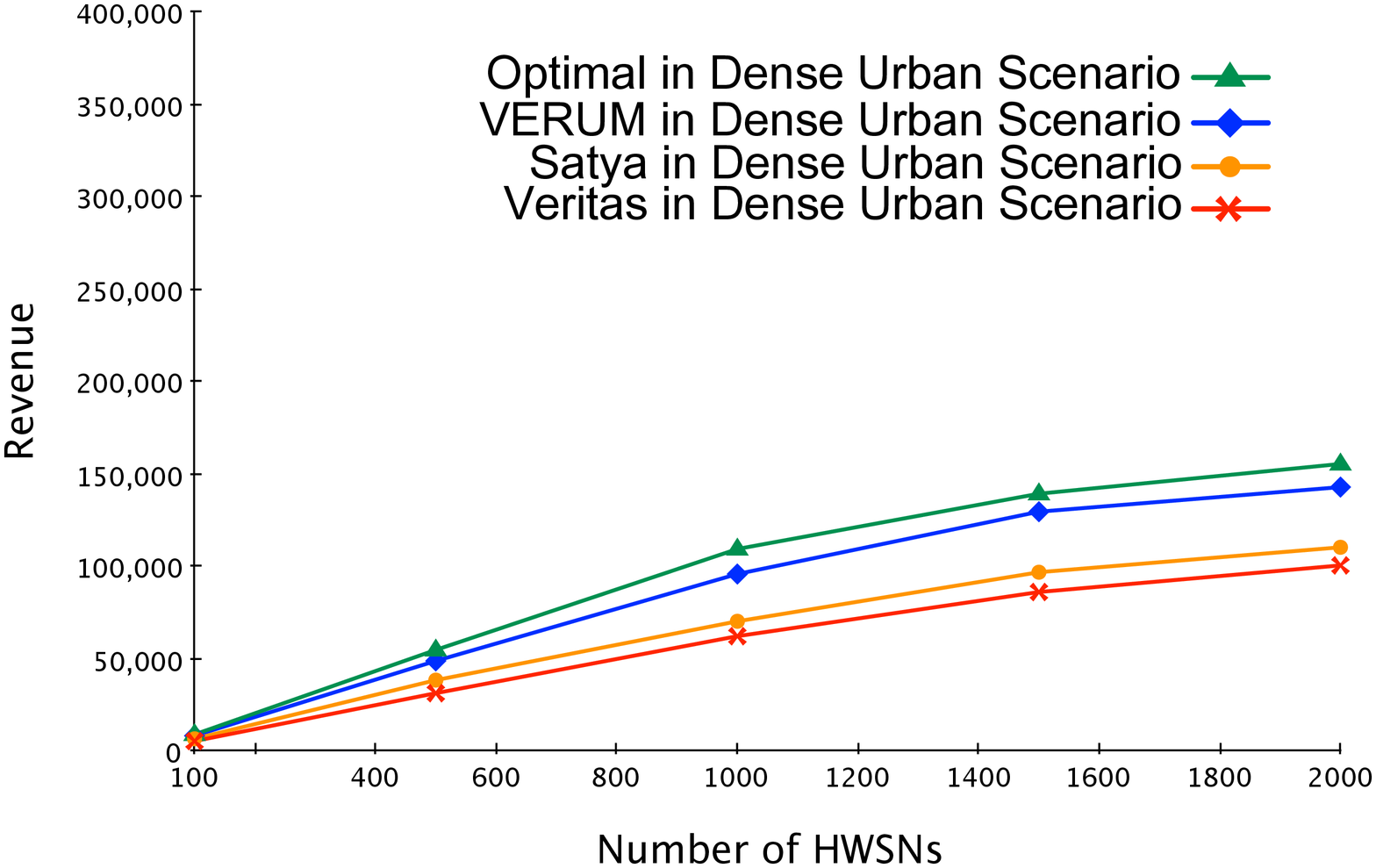}
 }
 \\
   \subfigure[]
{
	\includegraphics[width=2in]{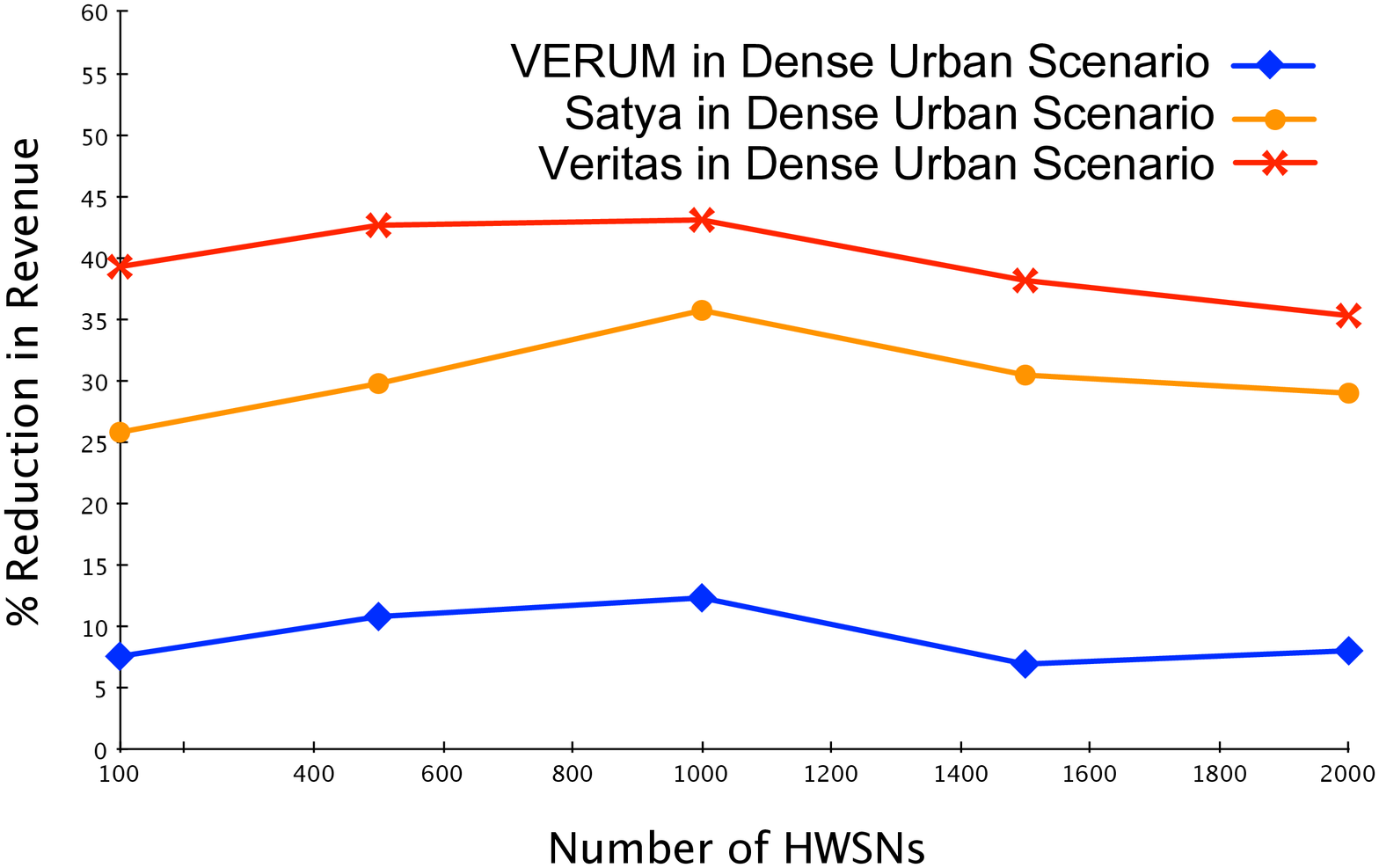}
 }
\\
    \subfigure[]
{
	\includegraphics[width=2in]{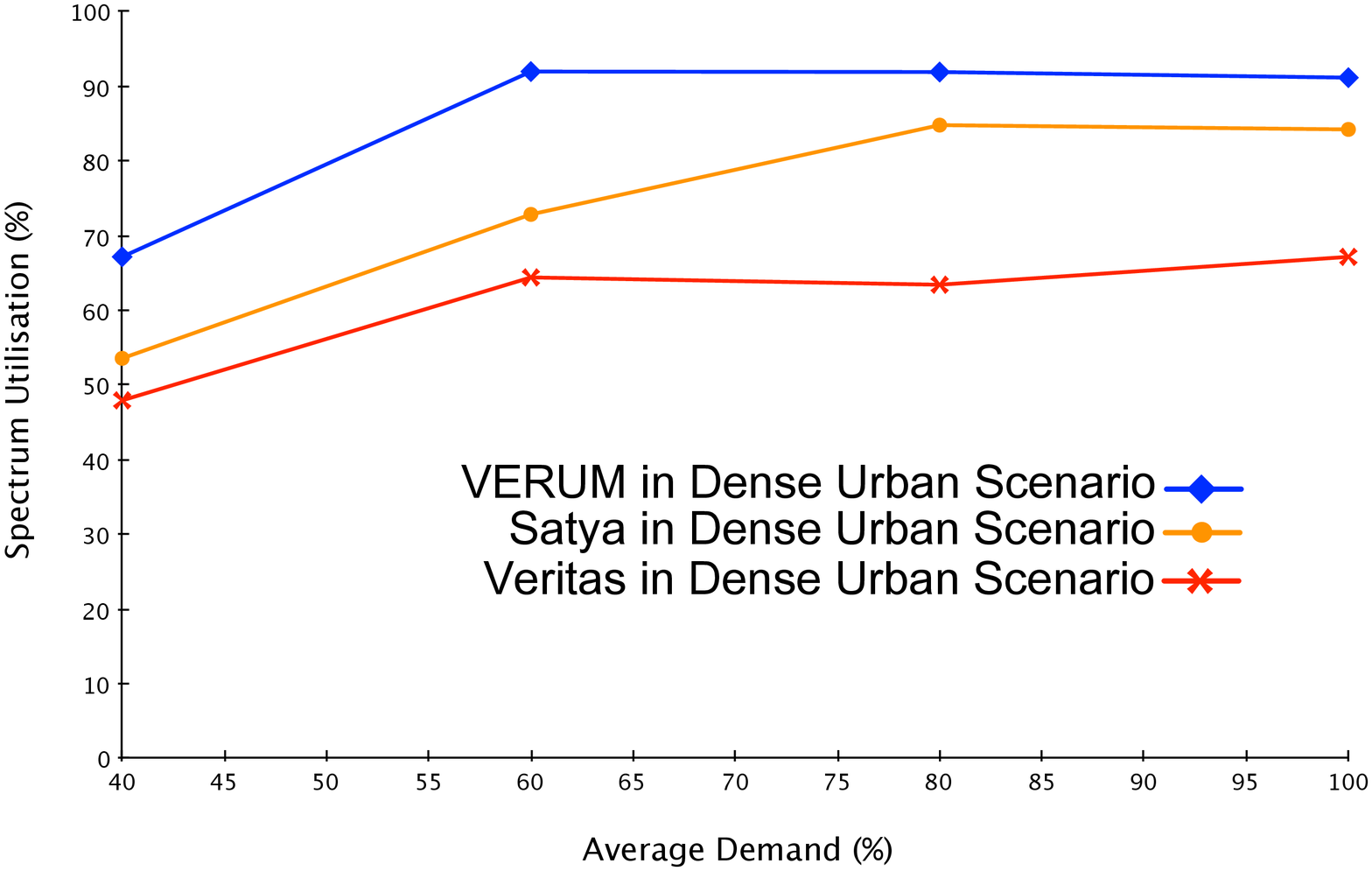}
 }
 \end{tabular}
\end{center}
 \caption{Revenue and spectrum utilization with varying number of HWSNs in the dense-urban scenario and 40\% third party WSNs.}
  \label{third-revenue}
\end{figure}

Fig.~\ref{third-revenue} (a) shows the absolute revenue in virtual currency units for different auction mechanisms including the optimal in the dense-urban scenario with varying number of active/subscribed HWSNs and 40\% third party WSNs. Fig.~\ref{third-revenue} (b) shows the percentage reduction in revenue for \texttt{VERUM}, SATYA and VERITAS with respect to the optimal. We see that the revenue with \texttt{VERUM} is fairly close to the optimal, within 10\% of the optimal in most cases. On the other hand, SATYA's performance is only slightly better than VERITAS but much worse than \texttt{VERUM}. Heterogeneous spectrum availability necessitates a higher dependance of channel sharing opportunities on the dividing HWSNs into buckets and ordering of buckets in SATYA. The ordering schemes used in SATYA are too simple to exploit all the available channel sharing opportunities. For the same reason, \texttt{VERUM} outperforms SATYA and VERITAS in terms of spectrum utilization too as seen from Fig.~\ref{third-revenue}(c). For instance, at 40\% demand \texttt{VERUM} improves spectrum utilization by around 30\% compared to SATYA and even more relative to VERITAS.

 \begin{figure}[h!]
  \begin{center}
  \begin{tabular}{c}
   \subfigure[]
{
	\includegraphics[width=2in]{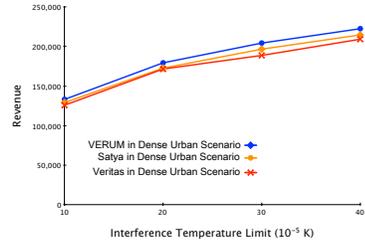}
 }
\\
   \subfigure[]
{
	\includegraphics[width=2in]{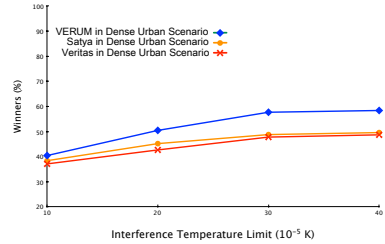}
 }
 \\
 \subfigure[]
{
	\includegraphics[width=2in]{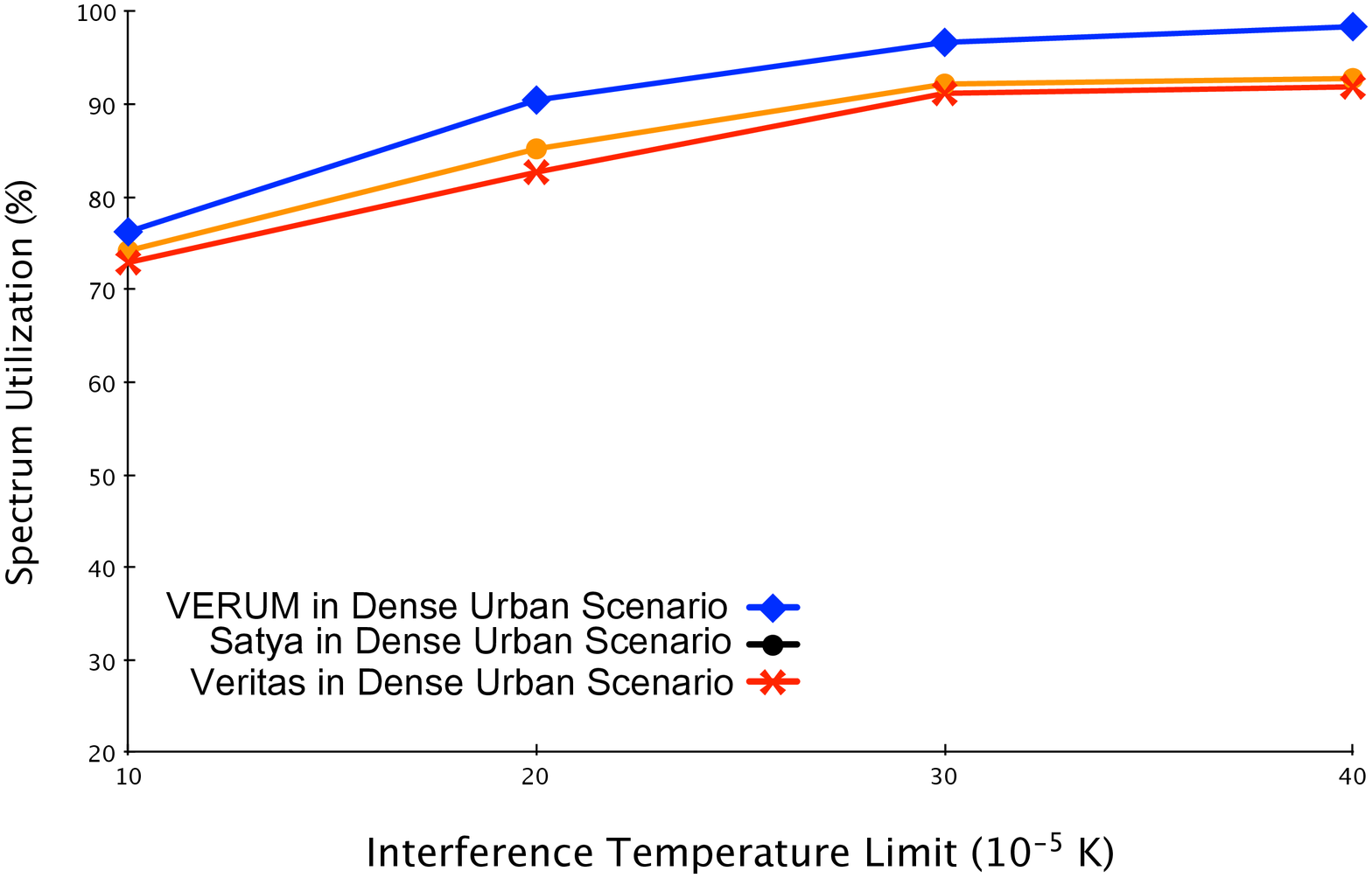}
 }
 \\
 \end{tabular}
\end{center}
 \caption{Varying Interference Temperature Limit}
  \label{int-temp-vary}
\end{figure}

\subsection{Additional Evaluation Results}
\label{additionalResults}
 \begin{figure*}[th]
  \begin{center}
  \begin{tabular}{cccc}
   \subfigure[]
{
	\includegraphics[width=1.5in]{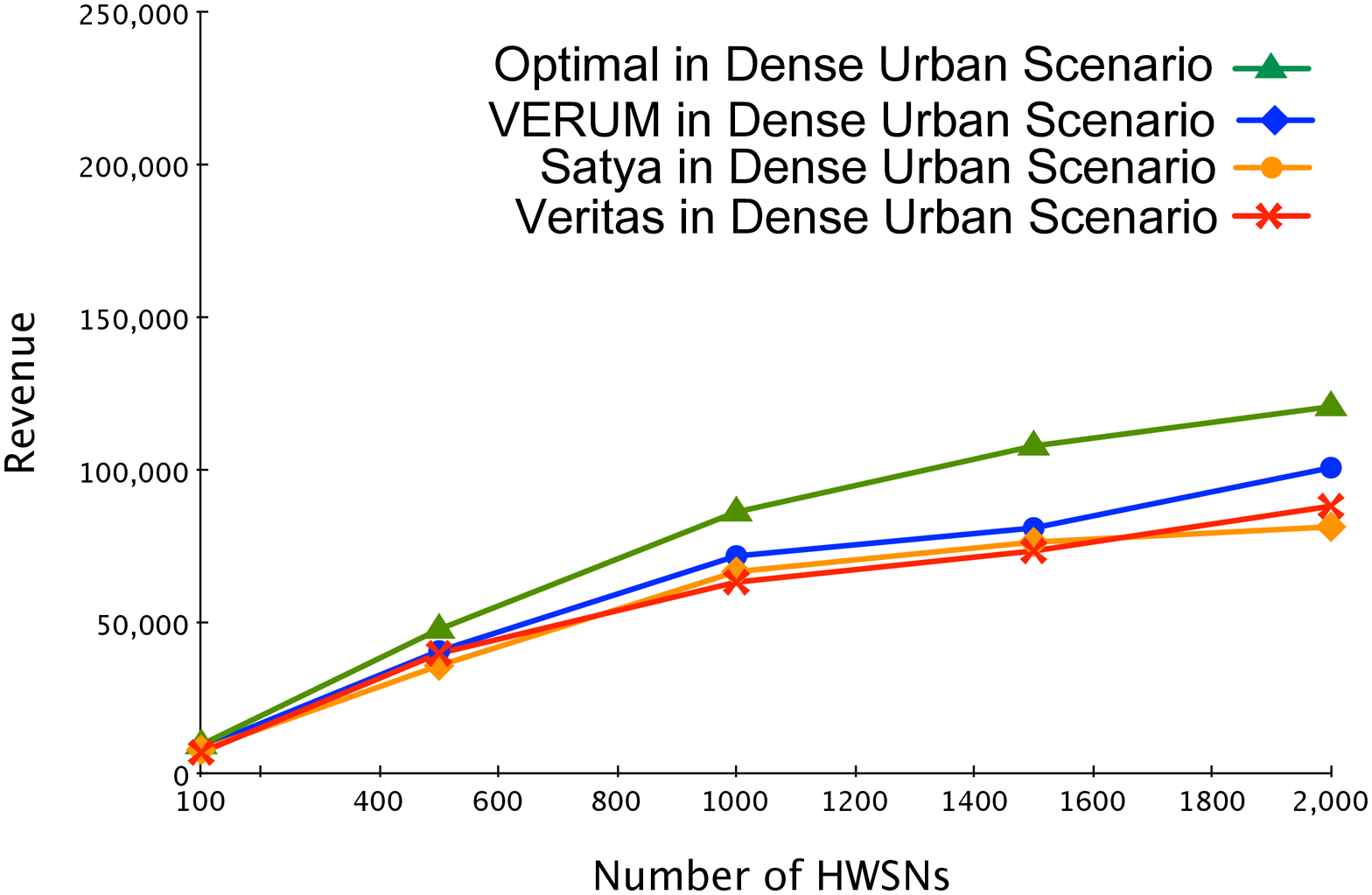}
 }
 &
   \subfigure[]
{
	\includegraphics[width=1.5in]{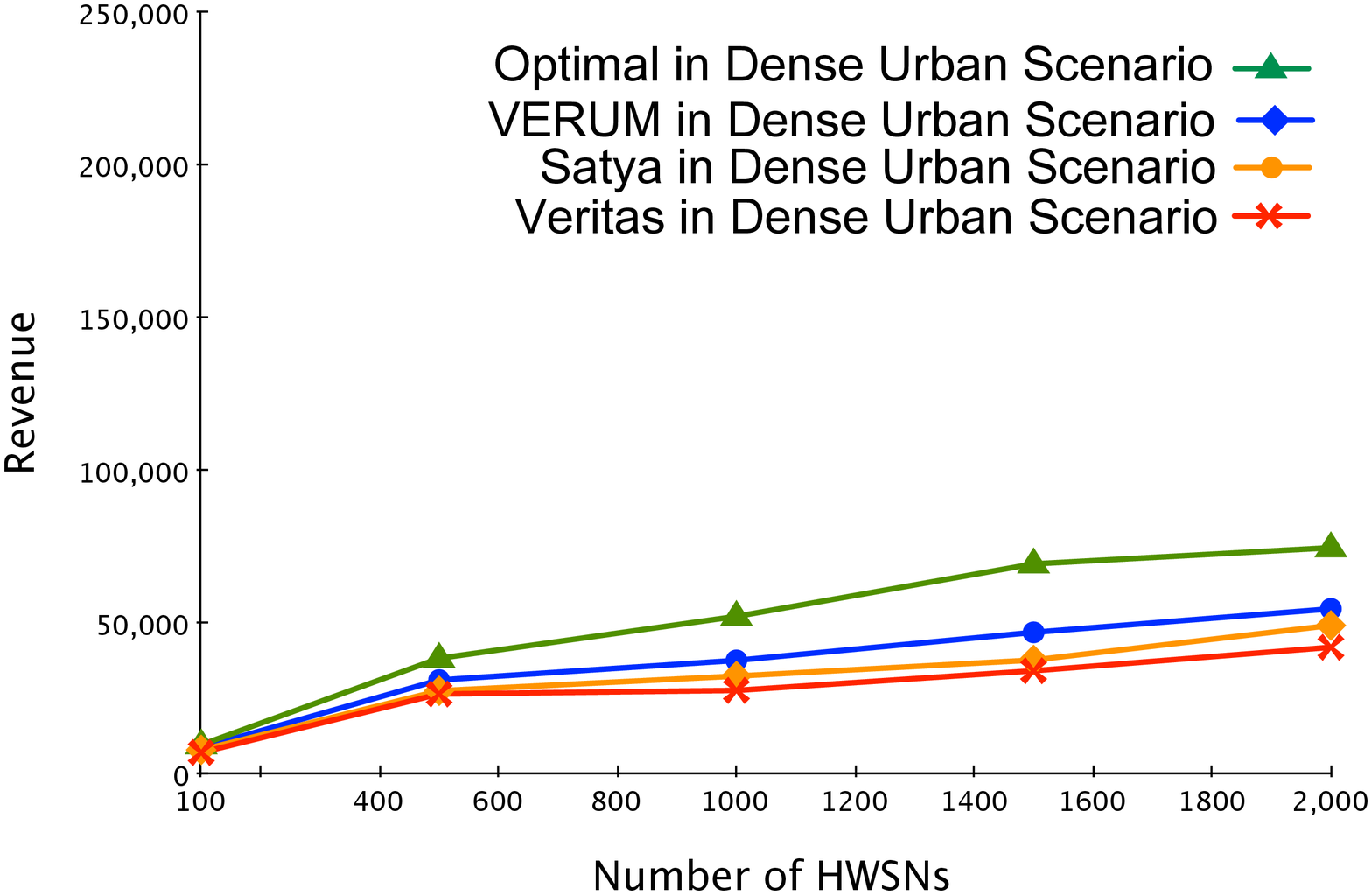}
 }
&
 \subfigure[]
{
	\includegraphics[width=1.5in]{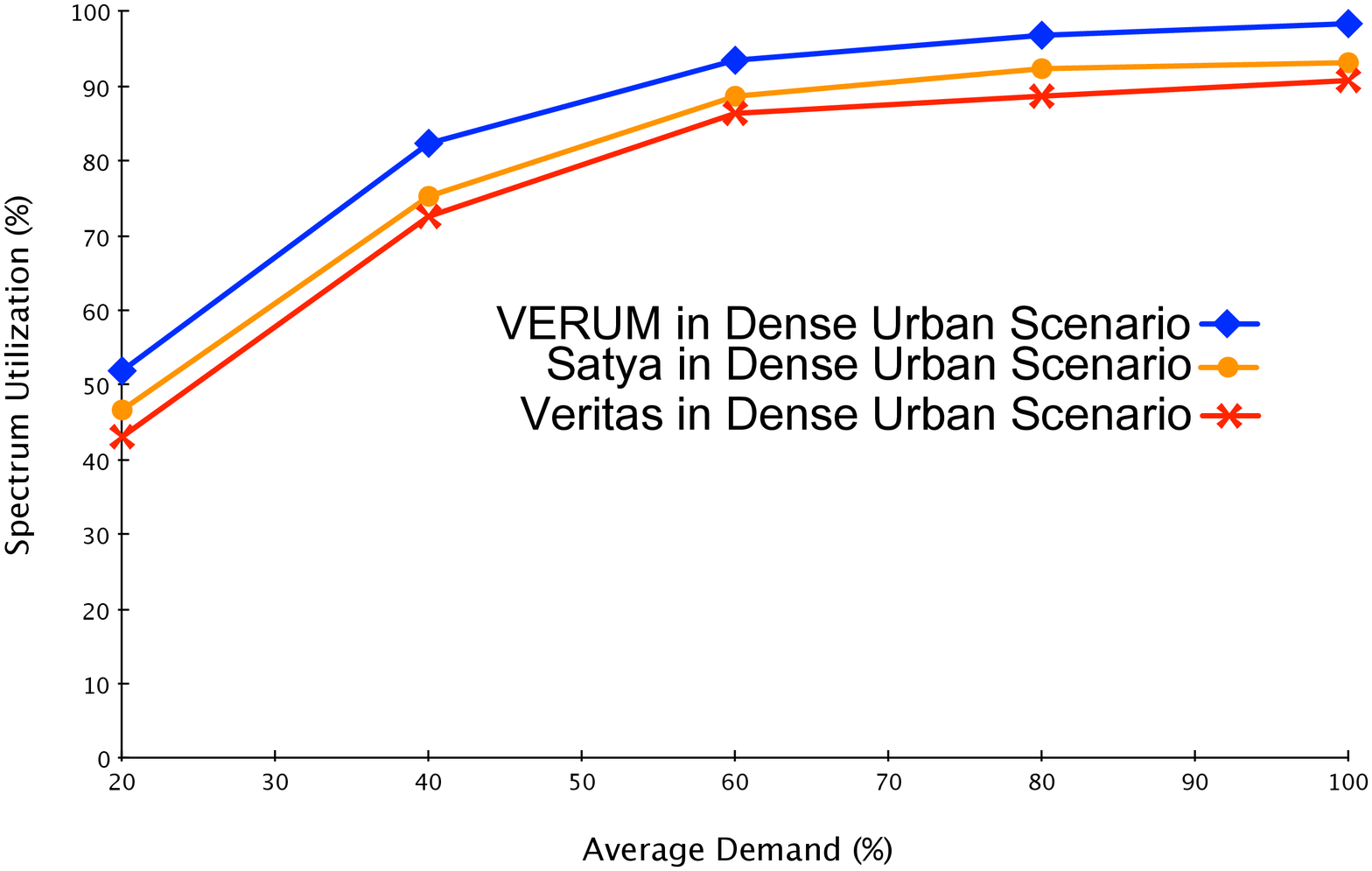}
 }
&
   \subfigure[]
{
	\includegraphics[width=1.5in]{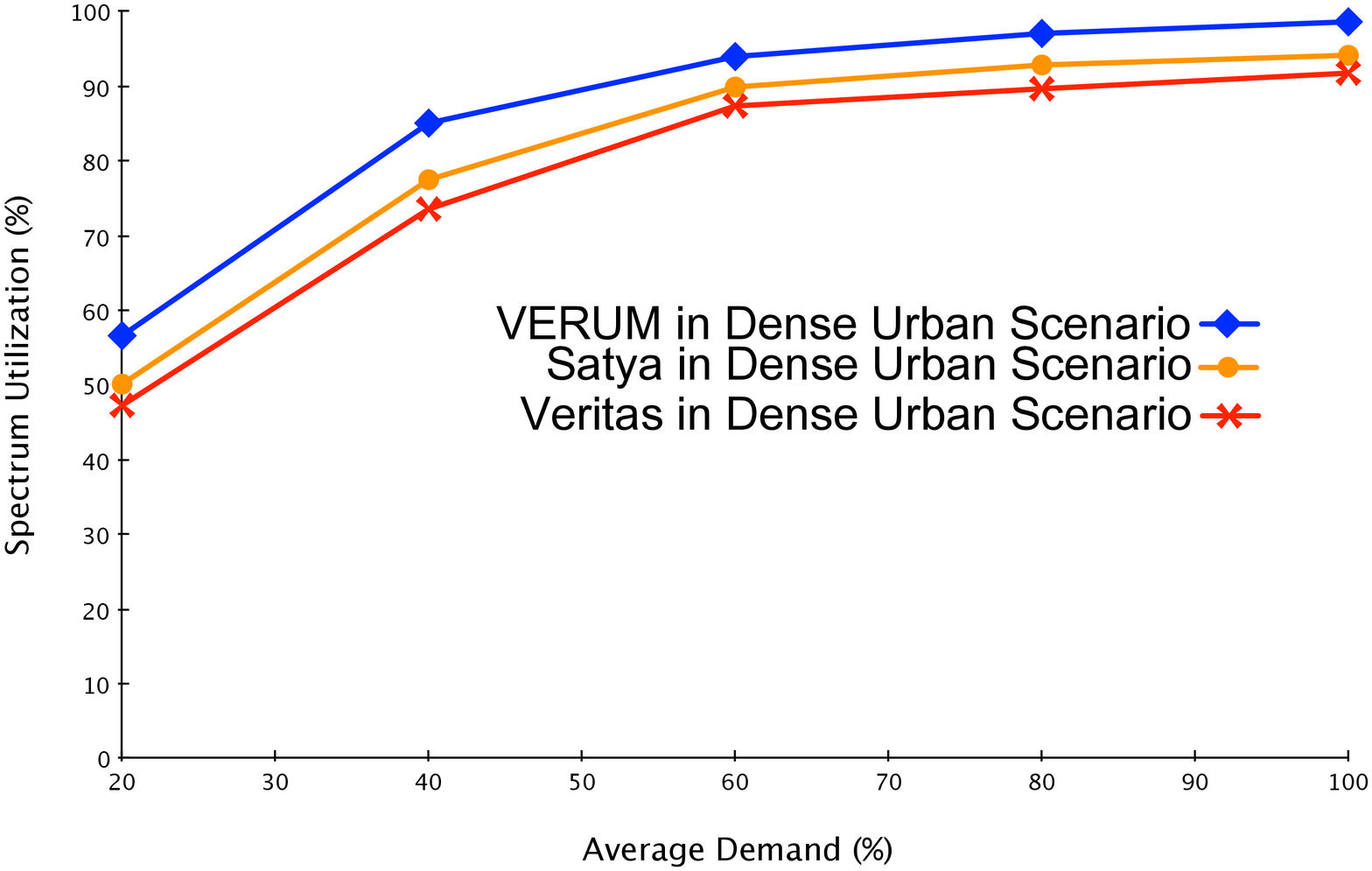}
 }
 \\
 \subfigure[]
{
	\includegraphics[width=1.5in]{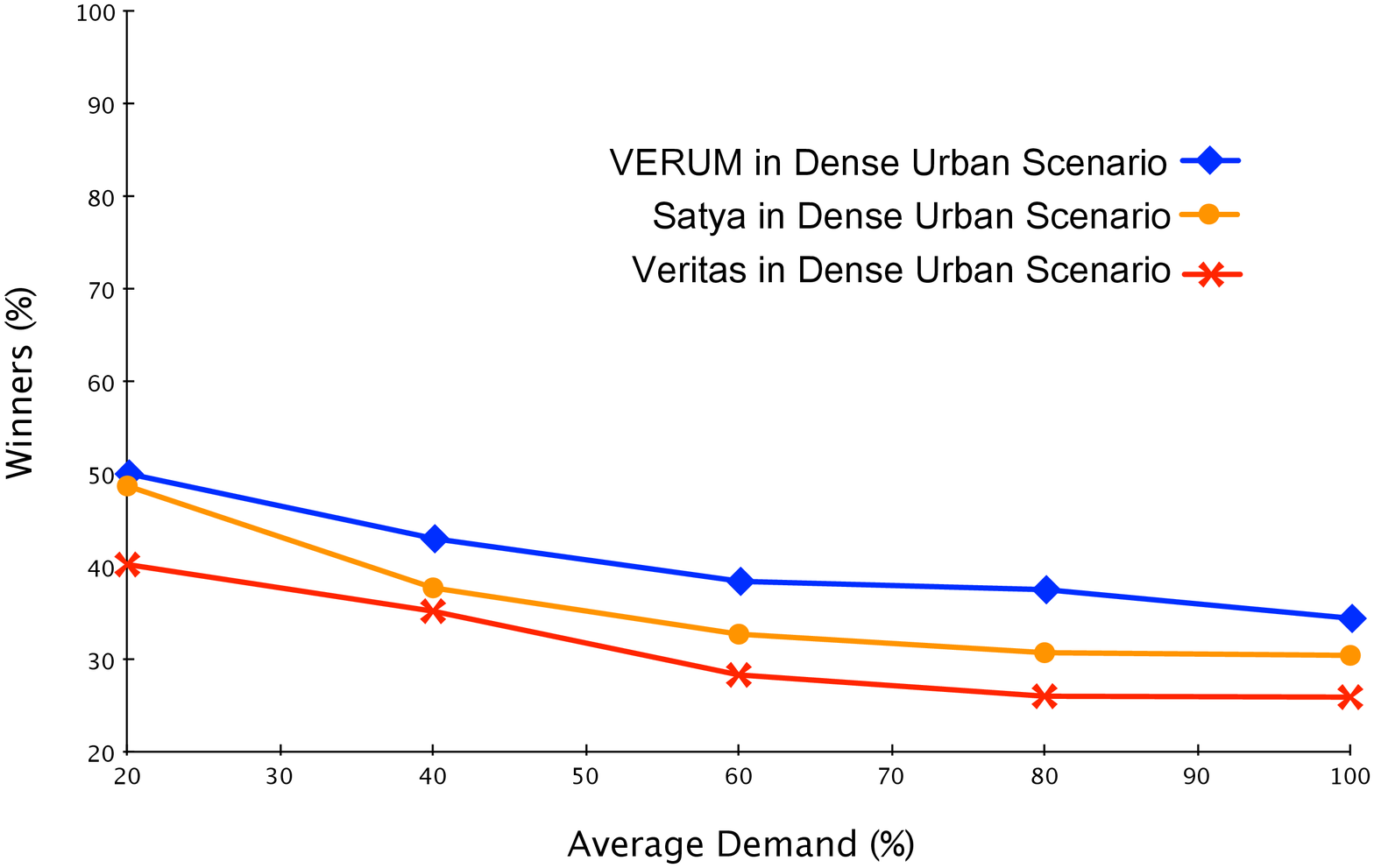}
 }
 &
   \subfigure[]
{
	\includegraphics[width=1.5in]{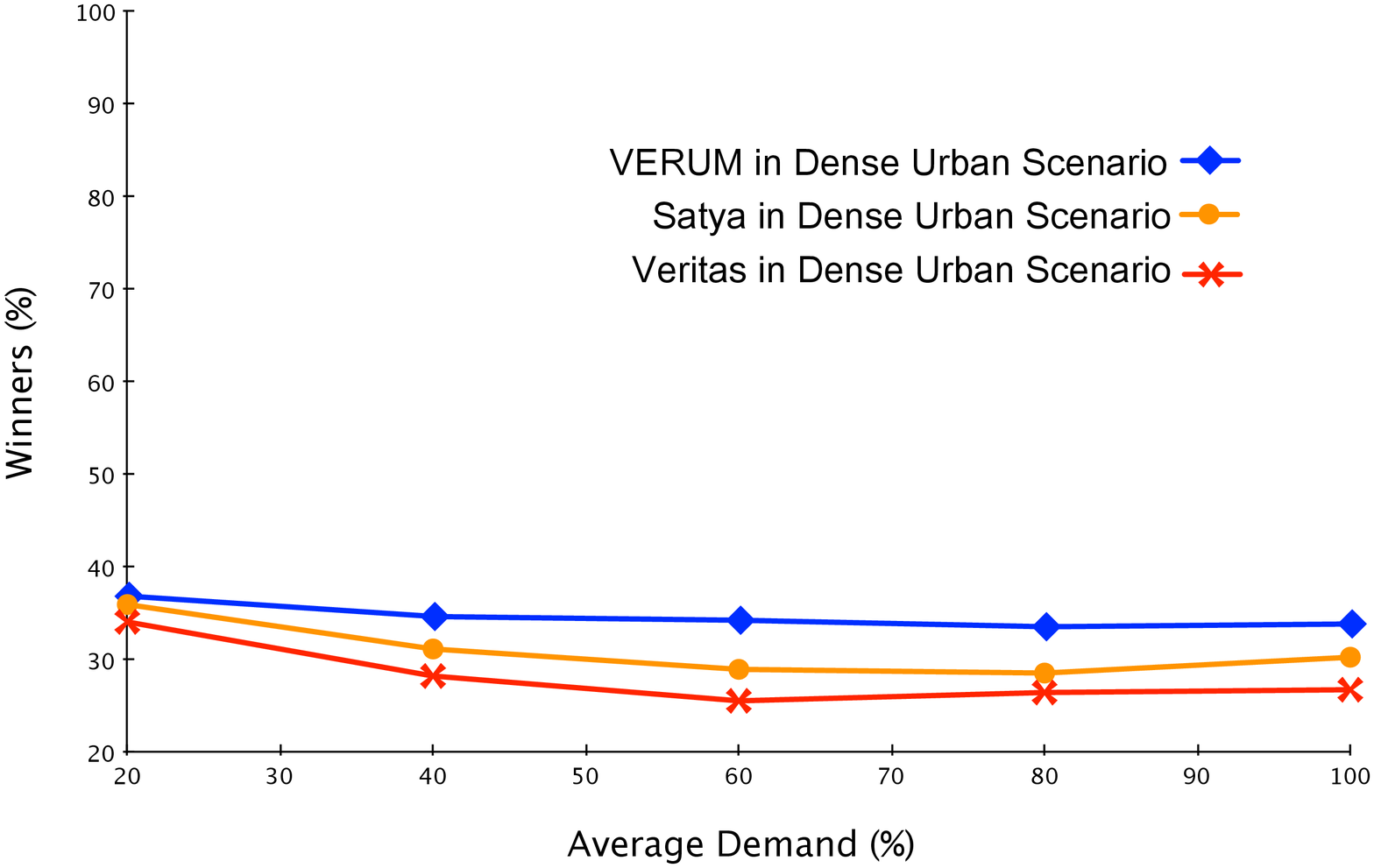}
 }
  &
 \subfigure[]
 {
	\includegraphics[width=1.5in]{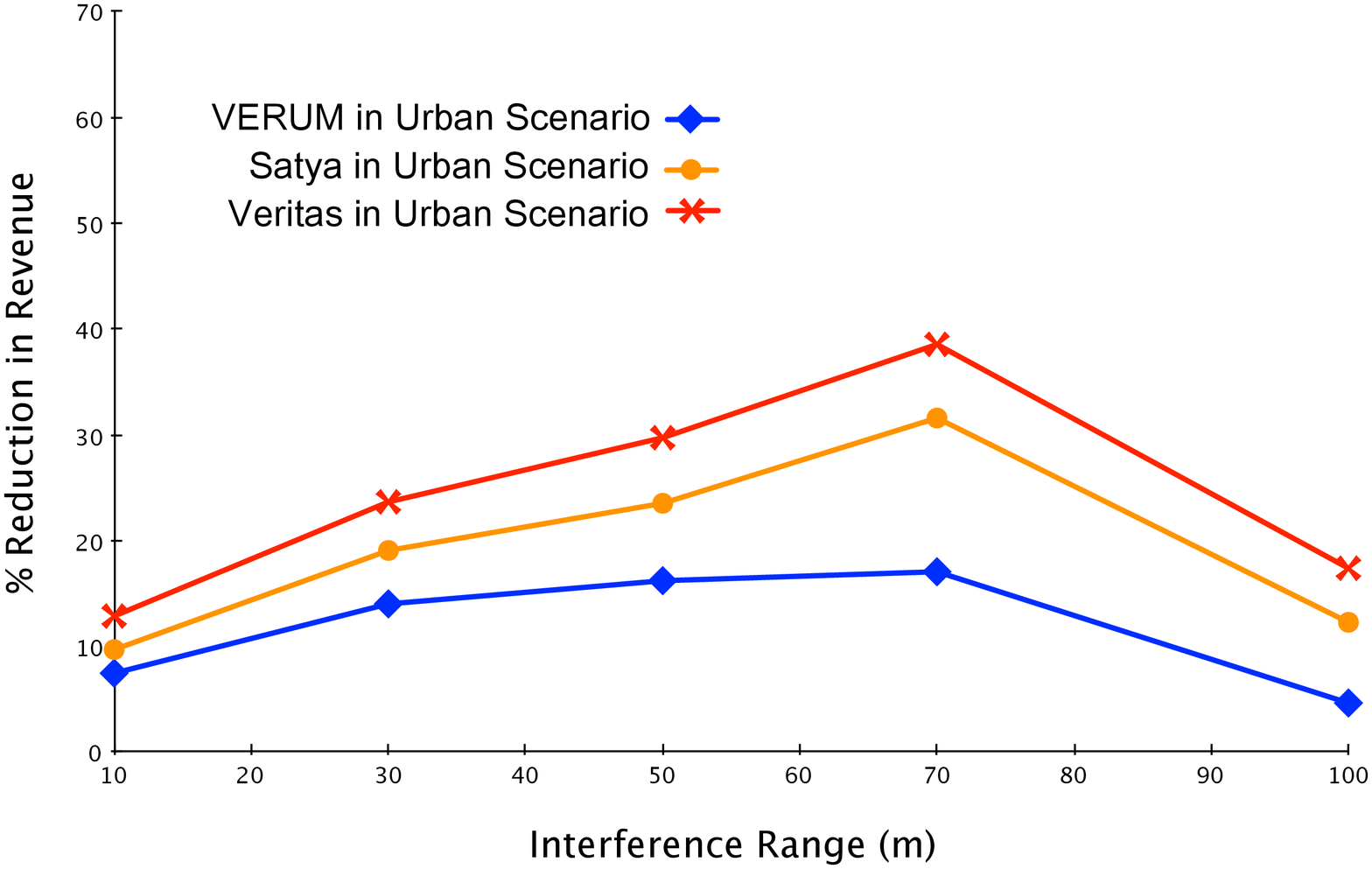}
        	\label{revenue_demand}
 }
 &
   \subfigure[]
{
	\includegraphics[width=1.5in]{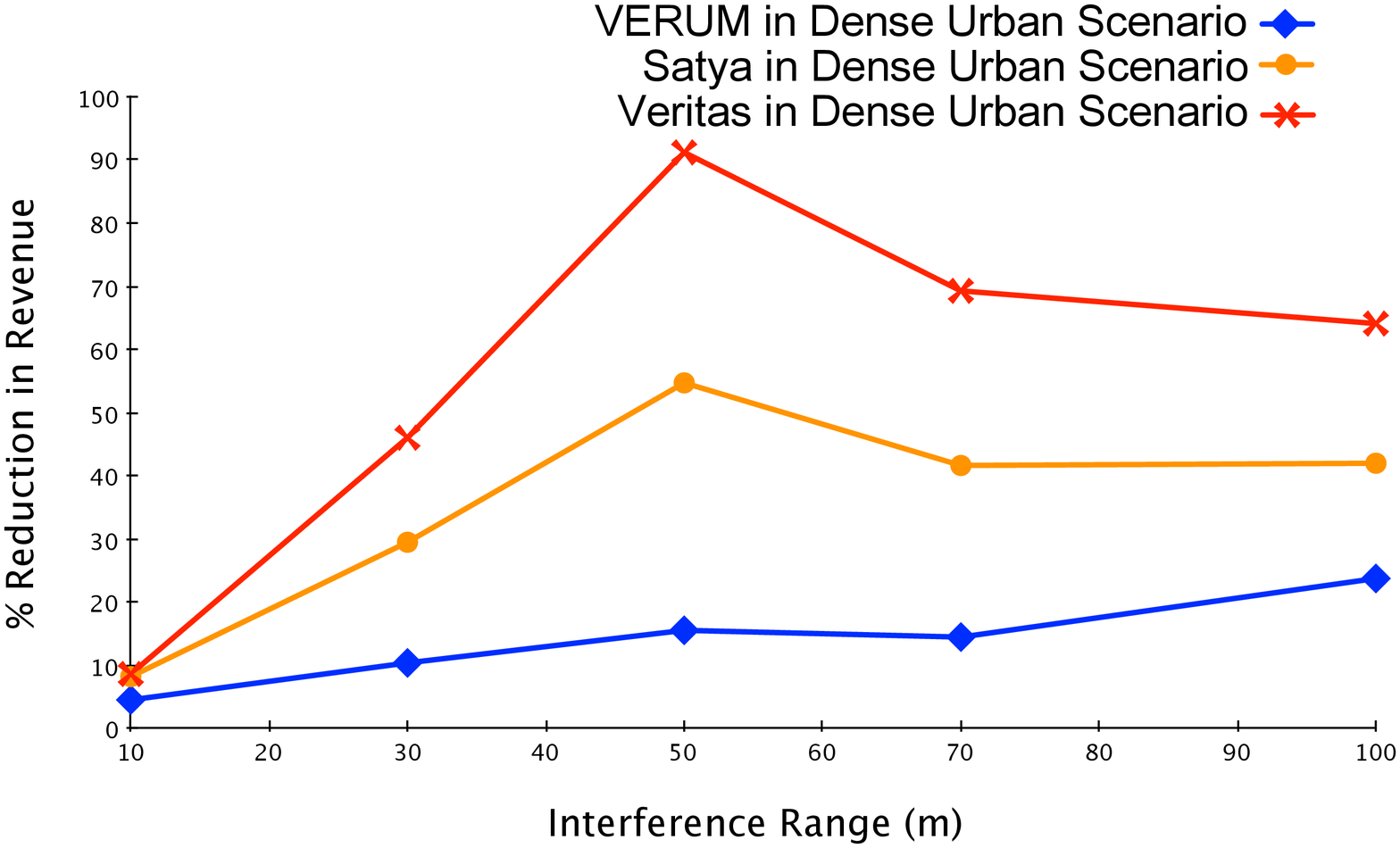}
        	\label{revenue_demand_hd}
 }
 \end{tabular}
\end{center}
 \caption{Varying Interference Range: Figures a, c, and e corresponds to 50 m, and Figures b, d, f to 70 m}
  \label{int-range-vary}
\end{figure*}

\subsubsection{Varying the Interference Range}

So far the it has been assumed that the interference range of the WSDs are fixed at a distance of 30 meters. Keeping the values of the other parameters constant we vary the interference range of WSDs to 50m and 70m and the performance of VERUM is shown in Fig. \ref{int-range-vary}. Increasing the interference range of the WSDs reduces the degree of channel sharing possible in the network. The consequence of lower channel reuse can be seen in the significant reduction in revenue generated with higher interference ranges.  The number of winners also decreases with higher interference ranges due to the same reason.

Although the total number of opportunities for channel reuse is lower, the average spectrum utilization in the network does not change significantly. This is because, the channel is still considered to be in use, although the number of WSDs using the channel is lower. This shows the need to keep the interference range of the WSDs minimal.

\subsubsection{Varying the Interference Temperature}

We effectively control the average number of neighbouring WSDs sharing a channel, by varying the interference temperature limit. The performance of VERUM with varying interference temperature is shown in the Fig.\ref{int-temp-vary}. The revenue generated increases with increasing interference temperature limit as more WSDs are allowed to share the same channel. It can be seen that the revenue increase is significant because with a lower interference temperature limit, the only sharing opportunities lost are the ones due to the interference temperature. This means that bandwidth is still available in that channel and relaxing the interference temperature limit would enable the WSD to use that channel thereby increasing revenue.

A WSD may share a channel with its neighbour if two conditions are met, a) it does not violate the interference temperature limit b) capacity available in the channel. At lower interference temperature limits, even though there is capacity available in the channels, some sharing opportunities are not utilized as it would violate the interference temperature limit. This can be seen as the increase in percentage of winners with increasing interference temperature limit. Spectrum utilization also increases with higher limits as WSDs who were prohibited from sharing can now share the channel with its neighbours.

\section{Discussion}
\label{discuss}


For simplicity of exposition, we have assumed that channel lease periods are identical across all HWSNs although it is straightforward to accommodate variable lease times. Extension to allow for varying leasing periods can be done similarly to the "stickiness" concept used in \cite{buddhikot-dyspan05}. Specifically, each channel can be associated with a minimum period of time T units, but HWSNs are allowed to request each channel for multiples of time T. However, longer duration allocations prohibit spectrum access by other HWSNs. HWSNs, therefore, should be allowed prolonged spectrum access only if they pay proportionately higher price. The per channel reserve price is set by the auctioneer independently for each HWSN depending on the requested leasing period. Auctioneer will determine the spectrum availability for secondary access at each epoch by additionally taking into account previous allocations that span multiple lease periods. Since our mechanism has a built-in discriminatory pricing model, the implementation of this feature is trivial.

Secondly, \texttt{VERUM} is a type of multi-unit auction scheme. While it considers all the channels as interchangeable goods, in reality this is not the case. The characteristics of a channel could vary depending on the physical location of the user. In order to capture this it is not essential that the auction has to be modeled as a combinatorial auction. By providing the desired channel characteristics as an input to computing the channel opportunity factor a similar outcome can be achieved. 

Moreover, we have assumed that all channels have the same fixed power limits. This however, is not realistic, since the power limits depend on the distance from the TV-stations and the environment. Moreover, two channels with different power limits can support different data rates. The higher data rate channels are considered of higher quality and, following the modern market economy, they should be more expensive. To capture this, but at the same time maintain the simplicity of our approach, channels can be categorized into distinct classes each associated with a distinct quality. The reserve prices of the channels within each class can then be determined based on the quality of the channels within the class (i.e., the higher the class, the higher the quality). 

Also related to this is the interference model. We assume a fixed interference range in the evaluation of the proposed scheme. In reality, this is almost never the case. One of the primary function of the the TVWS database provider is to implement an accurate interference model, which could be based on a propagation model or on empirical data. This accurate propagation model is used to create the conflict graph which is provided as an input to the proposed scheme. The discussion of an accurate interference model is out of scope of the paper and we refer to \cite{propagation-dyspan08} for a discussion on the techniques at the TVWS database provider. 

In practice, conflict graphs could be computed based on propagation models, measurement calibrated propagation models, or measurements. While propagation based models are the simplest, they are also the least accurate. Measurement based models use exhaustive measurements covering all possible sender and receiver locations. The measurement calibrated propagation model removes the need for exhaustive measurements by interpolating the signal strengths using calibrated models but still generating accurate conflict graphs will low error percentage. Zhou et al \cite{zhou-sigmetrics13} propose one such model using a graph augmentation technique that also address the aggregate interference in the network.

\section{conclusions}

In this paper, we have considered the problem of TVWS spectrum sharing among secondary users with home wireless networking as the motivating use case. We have  presented an interference-aware coordinated TVWS spectrum sharing framework for home networks that relies on short-term auctions and leverages the geolocation database to additionally keep track of secondary use of TVWS spectrum. To enable this auctioning based coordinated sharing framework, we have outlined a market-driven TVWS spectrum access model. For short-term auctions, we have developed an online multi-unit auction mechanism called \texttt{VERUM} that is shown to be truthful and efficient.

We evaluate \texttt{VERUM} using real home distributions in urban and dense-urban residential scenarios in London, UK in conjunction with realistic TVWS spectrum availability maps for the UK. Our evaluations show that \texttt{VERUM} performs close to optimal in terms of revenue even though it does not explicitly optimize revenue. Evaluations also show that \texttt{VERUM} outperforms VERITAS and SATYA -- the two state-of-the-art truthful and efficient multi-unit auction schemes -- in terms of revenue, spectrum utilization and percentage of winning bidders.

These results demonstrate that \texttt{VERUM} offers an effective alternative to uncoordinated TVWS spectrum use for home networking applications with incentives for both subscribed users of the auctioning based coordination service as well as for the auctioneer (spectrum manager).

\bibliographystyle{unsrt}
\bibliography{paper}  

\end{document}